%% file: main.tex
\documentclass[11pt]{article}

\usepackage[margin=1in]{geometry}
\usepackage[T1]{fontenc}
\usepackage[utf8]{inputenc}
\usepackage{lmodern}
\usepackage{microtype}
\usepackage{amsmath,amssymb,amsthm,mathtools}
\usepackage{booktabs,tabularx,array,multirow}
\usepackage{xcolor}
\usepackage{enumitem}
\usepackage{placeins}
\usepackage{xspace}
\usepackage{tikz}
\usetikzlibrary{arrows.meta,positioning,fit,quotes,calc,shapes.geometric}
\usepackage[hidelinks]{hyperref}
\usepackage[nameinlink,noabbrev]{cleveref}

\hypersetup{
  pdftitle={Beyond Distance Ordering: Resource Complexity and Universal Optimality of Exact Labeled Directed Shortest Paths},
  pdfauthor={Bin Cai}
}

\newtheorem{theorem}{Theorem}[section]
\newtheorem{lemma}[theorem]{Lemma}
\newtheorem{corollary}[theorem]{Corollary}
\newtheorem{proposition}[theorem]{Proposition}
\theoremstyle{definition}
\newtheorem{definition}[theorem]{Definition}

\newtheorem{remark}[theorem]{Remark}
\newtheorem{openproblem}[theorem]{Open Problem}

\crefname{theorem}{theorem}{theorems}
\Crefname{theorem}{Theorem}{Theorems}
\crefname{lemma}{lemma}{lemmas}
\Crefname{lemma}{Lemma}{Lemmas}
\crefname{corollary}{corollary}{corollaries}
\Crefname{corollary}{Corollary}{Corollaries}
\crefname{proposition}{proposition}{propositions}
\Crefname{proposition}{Proposition}{Propositions}
\crefname{definition}{definition}{definitions}
\Crefname{definition}{Definition}{Definitions}
\crefname{figure}{figure}{figures}
\Crefname{figure}{Figure}{Figures}
\crefname{equation}{equation}{equations}
\Crefname{equation}{Equation}{Equations}
\crefname{openproblem}{open problem}{open problems}
\Crefname{openproblem}{Open Problem}{Open Problems}

\newcommand{\R}{\mathbb{R}}
\newcommand{\Nzero}{\mathbb{N}_0}
\newcommand{\cR}{\mathcal{R}}
\newcommand{\CA}{\ensuremath{\mathrm{CA}}\xspace}
\newcommand{\DIST}{\ensuremath{\mathrm{DIST}}\xspace}
\newcommand{\DO}{\ensuremath{\mathrm{DO}}\xspace}
\newcommand{\SPT}{\ensuremath{\mathrm{SPT}}\xspace}
\newcommand{\rhoF}{\ensuremath{\rho_{\mathrm{fwd}}}\xspace}
\newcommand{\OPT}{\ensuremath{\operatorname{OPT}_{\mathrm{DIST}}}\xspace}
\newcommand{\WIN}{\ensuremath{\operatorname{WIN}}}
\newcommand{\Ready}{\ensuremath{\operatorname{Ready}}}

\newcommand{\spanof}{\operatorname{span}}
\newcommand{\ceil}[1]{\left\lceil #1\right\rceil}
\newcolumntype{Y}{>{\raggedright\arraybackslash}X}

\title{Beyond Distance Ordering:\\
Resource Complexity and Universal Optimality of\\
Exact Labeled Directed Shortest Paths}
\author{
  Bin Cai\\
  Wuhan University of Technology\\
  Wuhan, China\\
  \texttt{369028@whut.edu.cn}
}
\date{}

\begin{document}
\maketitle

\input{abstract}
\input{introduction}
\input{model}
\input{exact_additions_main}
\input{cyclic_pareto_main}
\input{scalar_benchmark}
\input{universal_main}
\input{efficient_main}
\input{related_work}
\input{discussion}
\input{limitations}
\input{conclusion}

\FloatBarrier
\clearpage
\appendix
\input{topology_parameter}
\input{addition_coordinate}
\input{dag_frontier}
\input{cyclic_atom}
\input{shared_hub}
\input{pareto_lower_bound}
\input{bbwo}
\input{rolling_block}
\input{pareto_law}
\input{universal_optimality}
\input{efficient_universality}
\input{appendix_new}
\input{appendix}

\FloatBarrier
\bibliographystyle{alpha}
\bibliography{references_core,references_flagship}

\end{document}

%% file: abstract.tex
\begin{abstract}
We study exact single-source shortest paths when the output is only the
materialized labeled distance vector (\DIST), rather than a distance order.
In the full deterministic comparison--addition model, the minimum worst-case
number of additions on every fixed directed topology is exactly the maximum
number $\rhoF$ of forward nonsource endpoint classes over rooted vertex
orders; the lower bound permits adaptive control, literals, and arbitrary
mixed sums.  This arithmetic law aligns with the comparison optimum on DAGs,
where the full resource region is an exact rectangle.  Cycles destroy that
alignment: a two-spoke shared-hub graph has coordinatewise optima $(4,2)$ but
requires five comparisons at the two-addition budget.  Its $k$-spoke
extension forces $k\log_2k+O(k)$ comparisons at the addition optimum and has
an entropy-tight deterministic tradeoff
$C_{k+r}^*(H_k)=\Theta(k+\Lambda_{k,r})$, where
$\Lambda_{k,r}=\log_2(k!/[r!(r+1)^{k-r}])$, with leading constant one when
$\Lambda_{k,r}/k\to\infty$.  Because the two coordinatewise minima need not
belong to one program, these conflicts lead to the same-program benchmark
$\OPT=\inf_A\sup_w(C_A(w)+P_A(w))$.  An exact transcript-cone game yields one
uniform interpreter whose charged addition--comparison cost equals $\OPT$ on
every topology; its optimal actions are synthesizable in polynomial space but
may require exponential time.  Finally, an active-core reduction and the
current deterministic directed-SSSP bound give an efficient uniform
$O\!\bigl(\OPT\sqrt{\log(2+\OPT)\log\log(4+\OPT)}\bigr)$ charged-operation
bound.  Thus optimal numerical policies exist uniformly, while efficient
constant-competitive navigation remains open.
\end{abstract}

%% file: introduction.tex
\section{Introduction}
\label{sec:introduction}

Shortest paths admit several output contracts.  A program can return the
labeled numerical vector, a nondecreasing order of vertices, or a tree of
witness paths.  These outputs are related, but they need not have the same
numerical complexity.  This paper isolates
\[
 \DIST(G,s)=\bigl(d(s,v)\bigr)_{v\ne s},
\]
with every coordinate held in an actual numerical register.  No settlement
order or path witness is requested.  Recent directed-SSSP algorithms below
the sorting barrier already make this distinction algorithmically relevant
\cite{duan2025breaking,duan2026faster}; universal-optimal Dijkstra results,
by contrast, benchmark the order-producing task \DO
\cite{haeupler2024universal,vanderhoog2025simpler}.

A common-suffix example makes the output boundary literal: two comparisons
determine all labeled distances, while producing their order requires
$\Theta(n\log n)$ comparisons (\cref{app:dist-do-separation}).  Our central
question is therefore:

\begin{quote}
\emph{If the output is only the exact labeled distance vector, what numerical
resources are intrinsically necessary, and can one uniform algorithm match
the best topology-specialized program?}
\end{quote}

We answer this in the deterministic comparison--addition (\CA) model.  A
program may adaptively compare any two held registers and may add any two
held finite values.  It may reuse registers, form mixed non-path sums, use
repeated operands and finitely many real literals, and share intermediate
values among output labels.  Copying and discrete control are free, but an
output value must be materialized.  The topology is fixed for a specialized
program; arc weights range over all nonnegative reals.

\paragraph{Exact arithmetic.}
Our first result characterizes the addition coordinate on every topology.
For a rooted order of the vertices, count the nonsource endpoint classes
directed forward, and let $\rhoF(G,s)$ be the maximum count.  Then
\[
 \boxed{P^*(G,s)=\rhoF(G,s).}
\]
Endpoint-class Dijkstra attains the value uniformly without computing the
NP-hard parameter.  The content is the lower bound: it applies to arbitrary
adaptive mixed-sum programs, rather than merely counting relaxations in a
particular shortest-path algorithm.  A generic potential-and-slack input,
the full path-difference kernel, and a projected fiber-opening argument force
one addition per forward endpoint class.

\paragraph{Acyclic alignment, cyclic conflict.}
On a reachable DAG with $n$ vertices, $m$ primitive arcs, $q$ endpoint
classes, and $d_s$ source-out classes, the entire resource region is
\[
 \boxed{\cR(G,s)=\{(c,p):c\ge m-n+1,\ p\ge q-d_s\}.}
\]
Thus the two resources align perfectly in the acyclic case.  We use this as
the baseline, not as the endpoint of the paper.

A minimal active cycle already breaks the rectangle.  For the two-spoke
shared-hub graph $H_2$,
\[
 \boxed{C^*(H_2)=4,\qquad P^*(H_2)=2,\qquad C_2^*(H_2)=5.}
\]
The first two values are coordinatewise optima, yet no program attains both.
The last inequality follows from a mathematical reduction to eight strict
cells and a finite exhaustive coefficient-state certificate.  We state this
computer-assisted boundary explicitly and provide a deterministic verifier.

\paragraph{Endogenous sorting and its tradeoff.}
The graph $H_k$ repeats the spoke around one hub.  It still requests only
labeled \DIST, and $P^*(H_k)=k$, but an addition-optimal program must
distinguish $k!$ scale-separated spoke orders:
\[
 \ceil{\log_2(k!)}\le C_k^*(H_k)\le S(k)+2k,
\]
where $S(k)$ is the optimal pairwise-sorting comparison count.  Hence
$C_k^*(H_k)=k\log_2 k+O(k)$ while $C^*(H_k)=\Theta(k)$.  This is endogenous
sorting in the information sense.  We do not claim that arbitrary aggregate
comparisons can be converted one-for-one into pairwise key comparisons.

The full interpolation exposes the mechanism.  Set
\[
 F_k(r)=C_{k+r}^*(H_k),\qquad
 \Lambda_{k,r}=\log_2\frac{k!}{r!(r+1)^{k-r}}
 \quad(0\le r\le k).
\]
We prove
\[
 \boxed{F_k(r)=\Theta(k+\Lambda_{k,r})}
\]
uniformly, with leading constant one on the entropy term whenever
$\Lambda_{k,r}/k\to\infty$.  The lower bound is causal.  After $t$ additions
at most $t$ disjoint spoke pairs can be exposed.  If the pair in position $s$
is still hidden after $s+r$ additions, a synchronized threat makes that pair
essential to $k-s+1$ distinct outputs although only $k-s$ additions remain.
The resulting deadline allows at most
$r!(r+1)^{k-r}$ permutations per transcript.  For the upper bound,
deterministic multiple selection produces ordered blocks of width $r+1$, and
a rolling computation spends a second addition only inside the crossing
block.  The standard deterministic multiple-selection theorem contributes a
lower-order entropy term; we therefore do not claim a uniform additive-$O(k)$
formula.

\paragraph{The same-program benchmark.}
The cyclic tradeoff is also the bridge to universal optimality: it shows that
the topology-specific comparator must price one program, not assemble a
fictitious algorithm from separately optimized coordinates.  In particular,
the $H_2$ nonrectangle rules out $C^*+P^*$ as a measure of total numerical
work because the two minima may belong to different programs.  We therefore define
\[
 \boxed{\OPT(G,s)=\inf_A\sup_{w\ge0}
        \bigl(C_A(G,s,w)+P_A(G,s,w)\bigr).}
\]
Both resources are charged to one execution of one program.  This is the
benchmark for universal optimality throughout the paper.

\paragraph{No numerical uniformity tax.}
We characterize $\OPT$ by a transcript-cone game.  A state records the
homogeneous coefficient vectors of held registers and the physical
comparison outcomes seen so far.  Those outcomes define a cone of consistent
inputs.  A state is terminal only when every label has one fixed held
register that is correct throughout the entire cone.  Addition and comparison
moves reproduce the physical program model exactly.  We prove
\[
 \WIN_G(\{0,e_1,\ldots,e_m\},\varnothing,K)
 \Longleftrightarrow \OPT(G,s)\le K.
\]
A canonical minimax policy therefore gives one uniform interpreter
$U_{\rm pspace}$ satisfying
\[
 \boxed{T_{U_{\rm pspace}}(G,s)=\OPT(G,s)}
\]
for every topology.  The interpreter reconstructs topology-specific optimal
actions from the graph; they are not nonuniform advice.  Its state and
depth-first policy search use polynomial space, but the ordinary time between
charged operations may be exponential.  Exact numerical universal optimality
is settled; efficient navigation is not.

\paragraph{The efficient boundary.}
A canonical active-core reduction deletes dominator-inactive arcs, reduces
source-parallel classes, and retains only endpoints of active nonsource arcs.
If $k$ is the number of those arcs and $\sigma$ is the source-class reduction
cost, then
\[
 |V(H)|\le2k+1,\quad |E(H)|\le3k,\quad
 \sigma+k\le2\OPT(G,s),
\]
and
\[
 \OPT(H,q)\le\OPT(G,s)
 \le\sigma+\OPT(H,q)\le2\OPT(G,s).
\]
Combining this same-input charge with the current deterministic directed-SSSP
algorithm of Duan, Mao, Shu, and Yin~\cite{duan2026faster} gives one explicit
efficient uniform algorithm using
\[
 O\!\left(\OPT\sqrt{\log(2+\OPT)\log\log(4+\OPT)}\right)
\]
charged operations after linear topology preprocessing.

\begin{table}[htbp]
\centering
\small
\begin{tabularx}{\textwidth}{@{}Yccc@{}}
\toprule
Result & statement type & uniformity & computation status \\
\midrule
$P^*=\rhoF$ & exact, all topologies & uniform upper & efficient \\
DAG rectangle & exact baseline & uniform & efficient \\
$H_2$ nonrectangle & exact & family-specific & lower endpoint computer-assisted \\
$H_k$ law & entropy-tight asymptotic & constructive upper & efficient \\
Transcript-cone theorem & exact, all topologies & ratio one & PSPACE, possibly exponential \\
Active-core lifting & sublogarithmic factor & uniform & efficient \\
\bottomrule
\end{tabularx}
\caption{The theorem hierarchy and scope.  ``Exact'' refers to charged
numerical operations, not ordinary running time.}
\label{tab:theorem-map}
\end{table}

\paragraph{Organization.}
\Cref{sec:model} fixes the output contract and program model.
\Cref{sec:main-additions,sec:main-dag} give the exact arithmetic law and its
acyclic control case.  \Cref{sec:main-cycles,sec:main-pareto} develop cyclic
conflict, endogenous sorting, and the Pareto law.
\Cref{sec:scalar-benchmark,sec:main-universal} introduce $\OPT$ and prove the
transcript-cone equivalence.  \Cref{sec:main-efficient} locates the efficient
boundary.  Detailed proofs and reproducibility material follow in the
appendices.

%% file: model.tex
\section{Model, output tasks, and benchmarks}
\label{sec:model}

\subsection{Rooted directed topology}

Let \(G=(V,E)\) be a finite directed multigraph with distinguished source
\(s\).  Every vertex is reachable from \(s\).  Each primitive arc
\(e\in E\) has an independent weight \(w_e\in\R_{\ge0}\).  Parallel
primitives, ties, zero weights, and zero-weight cycles are allowed.  A
self-loop is permitted syntactically but is benchmark-neutral: a nonnegative
self-loop never improves a distance, so it can be ignored in one direction
and replaced by the zero literal in the other.  Unreachable labels, if any,
are removed before the numerical instance is formed.

For \(v\in V\), nonnegativity gives a simple shortest-path representative,
and
\[
 d_w(s,v)=\min\left\{\sum_{e\in P}w_e:
                 P\text{ is a directed }s\text{--}v\text{ path}\right\}.
 \tag{10}\label{eq:distance}
\]

\paragraph{Output contracts.}
The task \DIST requires one existing finite numerical register equal to
\(d_w(s,v)\) for every labeled \(v\ne s\).  It does not require an order of
the distances, a settlement order, a predecessor tree, or witness paths.  The
task \DO additionally outputs a total vertex order nondecreasing in distance,
with a fixed rule for ties.  The task \SPT outputs witness predecessor arcs.
All complexity quantities below refer to \DIST unless explicitly subscripted
otherwise.

\subsection{Topology-specialized comparison--addition programs}

A deterministic topology-specialized \CA program is finite code with finite
topology-dependent discrete advice.  Its initial numerical registers are the
primitive weights and finitely many literal real constants.  The charged
operations are:

\begin{enumerate}[leftmargin=*,itemsep=2pt]
  \item a binary addition of two already materialized finite registers,
        creating a new materialized register; and
  \item a comparison of two already materialized finite registers, branching
        on \(<\), \(=\), or \(>\).
\end{enumerate}

The program may branch adaptively, use equality outcomes, add a register to
itself, retain or reuse arbitrary registers, form mixed or non-path sums,
compare aggregates, and share work across output labels.  Copying, selecting,
indexing, and discrete control are free only when they reveal no new numerical
information.  There is no free subtraction, numerical scalar multiplication,
arbitrary linear-form construction, bit inspection, or output of an
unevaluated expression.  An output is legal only if its numerical value is
already held in a register.

This model is intentionally more permissive than relaxation-only and
path-register models.  Every lower bound below therefore has to survive
mixed supports, constants, repeated operands, equality faces, and adaptive
register reuse.

\subsection{Split resources}

For a program \(A\) and a weighting \(w\), let
\(C_A(G,s,w)\) and \(P_A(G,s,w)\) be the comparisons and additions on that
execution.  Define
\[
 C_A(G,s)=\sup_{w\ge0}C_A(G,s,w),\qquad
 P_A(G,s)=\sup_{w\ge0}P_A(G,s,w).
 \tag{11}\label{eq:worst-costs}
\]
The coordinate optima are
\[
 C^*(G,s)=\inf_A C_A(G,s),\qquad
 P^*(G,s)=\inf_A P_A(G,s),
 \tag{12}\label{eq:scalar-optima}
\]
and for \(b\in\Nzero\) the budgeted comparison optimum is
\[
 \boxed{C_b^*(G,s)=
 \inf\{C_A(G,s):A\text{ computes \DIST and }P_A(G,s)\le b\}.}
 \tag{13}\label{eq:budget-optimum}
\]
The value is \(+\infty\) when the budget is infeasible.  The upward-closed
resource region is
\[
 \cR(G,s)=\{(c,p)\in\Nzero^2:
    \exists A\text{ computing \DIST with }C_A\le c,\ P_A\le p\}.
 \tag{14}\label{eq:resource-region}
\]
Since all finite costs are integers and a finite Bellman--Ford program exists,
the finite infima used in this paper are attained.  A coordinatewise optimum
does not imply that one program attains both coordinates.

\subsection{Total numerical benchmark}

For one program \(A\), define its worst-case total charged cost by
\[
 T_A(G,s)=\sup_{w\ge0}
   \bigl(C_A(G,s,w)+P_A(G,s,w)\bigr),
 \qquad
 \boxed{\OPT(G,s)=\inf_A T_A(G,s).}
 \tag{15}\label{eq:opt-dist}
\]
The supremum is taken after adding the two resources on the same execution.
In general \(\OPT\) is not \(C^*+P^*\), nor is it defined by adding two
separately optimized worst cases.  This same-program scalar benchmark is the
reference quantity in \cref{sec:universal-numerical,sec:efficient-universality}.

\paragraph{Uniform algorithms and RAM work.}
A uniform algorithm has one finite description and takes \((G,s)\) as input.
Its topology preprocessing, ordinary RAM work, memory, comparisons, and
additions are distinct resources.  ``Exact numerical universal optimality''
means equality with \(\OPT\) in charged operations, with no assertion about
the ordinary computation selecting those operations.  ``Efficient
universality'' additionally requires \(O(m+n)\) topology preprocessing and
RAM work bounded by \(O(m+n+C+P)\) on each query.

\subsection{Affine coefficient records}
\label{sec:coefficient-record}

On a fixed reached branch, every materialized register has the affine form
\[
 R(w)=\langle a,w\rangle+\alpha,
 \qquad a\in\Nzero^E,
\]
where \(\alpha\) belongs to the field generated by the program's literal
constants.  Primitive inputs have unit coefficient vectors, literals have
zero coefficient vector, and a charged addition adds the two records.
Comparisons affect control but not coefficients.  These records describe
arbitrary held aggregates; they do not impose a path-register normal form.

%% file: exact_additions_main.tex
\section{Exact addition complexity}
\label{sec:main-additions}

Parallel primitives with the same ordered endpoints form an \emph{endpoint
class}; write $\overline E$ for the nonempty classes.  A total order $\prec$
of $V$ is \emph{rooted} if $s$ is first and every other vertex has an
incoming class from an earlier vertex.  Define
\[
 \rhoF(G,s)=\max_{\prec\ \mathrm{rooted}}
 \bigl|\{(u,v)\in\overline E:u\ne s,\ u\prec v\}\bigr|.
\]
This parameter counts forward nonsource endpoint classes, not primitive arcs.
It is NP-hard to evaluate, but the algorithm below does not evaluate it.

\begin{theorem}[Exact addition law]
\label{thm:main-addition}
For every finite reachable directed multigraph with nonnegative real weights,
\[
 \boxed{P^*(G,s)=\rhoF(G,s).}
\]
One uniform endpoint-class version of Dijkstra attains this value on every
topology.
\end{theorem}

\paragraph{Upper bound.}
Reduce each parallel class to its minimum.  When Dijkstra settles $u$, create
$d(u)+\mu_{uv}$ only for a class whose tail is not $s$ and whose head is not
yet settled.  The settlement order is rooted, so the additions are exactly
the classes counted by its forward count and are at most $\rhoF$.  A weighting
realizes every rooted order, which also shows that this analysis is tight for
the algorithm.

\paragraph{Lower-bound architecture.}
Fix a maximizing rooted order and retain its forward classes, obtaining a
rooted spanning DAG $D$.  Algebraically independent vertex potentials plus
small independent slacks give one positive input on which every $s$--$v$ path
in $D$ ties at the potential of $v$.  On the reached adaptive branch, each
held register is an affine form with a nonnegative integral coefficient
vector.  Equality normals must span the full common-endpoint path-difference
kernel: otherwise a two-sided kernel perturbation preserves the branch but
changes the correct minimum at some label.  Projecting the coefficient forms
along this equality space reduces materialization to a support-fiber process.
A binary addition can open at most one new nonsource endpoint fiber, so all
$\rhoF$ forward fibers require distinct additions.  The argument permits
mixed sums, repeated operands, literals, equality branches, and reuse across
outputs.  The full kernel and fiber-opening proofs are in
\cref{sec:topology,sec:additions}.

Two points prevent common shortcuts.  First, the hard input is chosen after
the competing finite program is fixed.  The perturbation radius is smaller
than every reached strict margin, while its two signs lie in the nullspace of
every reached equality normal.  Adaptivity therefore does not let the program
change branches to avoid the kernel argument.  Second, the fiber count is
performed after projecting away precisely those equalities.  A register
containing many endpoint classes, or several copies of one primitive, is
still one coefficient point; joining two old fibers can open at most one new
fiber in the quotient.  The bound counts physical additions, not syntactic
path extensions.

Literal constants require no separate assumption.  Correct distance outputs
are positively homogeneous in the primitive weights.  Along a scaling ray,
each finite constant is lower order; after one finite transcript is fixed,
the leading coefficient records determine every nonidentity comparison.
Comparisons whose leading forms coincide remain forced equalities.  The
theorem therefore covers affine registers rather than only homogeneous
straight-line programs.

\section{The acyclic baseline}
\label{sec:main-dag}

Let $G$ now be a reachable DAG with $n$ vertices, $m$ primitive arcs, $q$
endpoint classes, and $d_s$ source-out classes.

\begin{theorem}[Exact DAG resource rectangle]
\label{thm:main-dag}
\[
 \boxed{\cR(G,s)=\{(c,p)\in\Nzero^2:
              c\ge m-n+1,\ p\ge q-d_s\}.}
\]
Thus one topological dynamic program simultaneously attains both coordinate
optima.
\end{theorem}

Class reduction uses $m-q$ comparisons, and selecting one of the incoming
class candidates at each nonsource vertex uses $q-(n-1)$ more.  The same
schedule performs one addition for every nonsource class, giving the corner
$(m-n+1,q-d_s)$.  For the lower bound, the potential-tie input above is now
available on the whole graph.  Equality-rank forces $m-n+1$ comparisons, and
the projected fiber argument forces $q-d_s$ additions on that same execution.
The detailed proof is in \cref{sec:dag-frontier}.

On DAGs the two resources are therefore perfectly aligned.  This is the
control case: the next section shows that a minimal active cycle destroys the
rectangle.

%% file: cyclic_pareto_main.tex
\section{From cyclic conflict to endogenous sorting}
\label{sec:main-cycles}

\subsection{A minimal nonrectangle}

Let $H_2$ have vertices $s,x,v,w$ and arcs
\[
 s\to x:a,\quad s\to v:b,\quad v\to x:z,\quad x\to v:u,
 \quad s\to w:c,\quad w\to x:q,\quad x\to w:p.
\]
Its distances are
\[
 d_x=\min\{a,b+z,c+q\},\quad
 d_v=\min\{b,d_x+u\},\quad d_w=\min\{c,d_x+p\}.
\]

\begin{figure}[htbp]
\centering
\begin{tikzpicture}[
  vertex/.style={circle,draw,thin,minimum size=7mm,inner sep=1pt},
  every edge/.style={draw,thin,->,>=Stealth},
  every edge quotes/.style={font=\small,fill=white,inner sep=1pt}]
  \node[vertex] (s) at (0,0) {$s$};
  \node[vertex] (x) at (2.6,1.15) {$x$};
  \node[vertex] (v) at (2.6,-1.15) {$v$};
  \node[vertex] (w) at (5.3,0) {$w$};
  \path (s) edge["$a$"] (x) edge["$b$"'] (v)
        (v) edge[bend left=17,"$z$"] (x)
        (x) edge[bend left=17,"$u$"'] (v)
        (s) edge[bend right=12,"$c$"'] (w)
        (w) edge[bend left=17,"$q$"'] (x)
        (x) edge[bend left=17,"$p$"] (w);
\end{tikzpicture}
\caption{The two-spoke shared-hub atom $H_2$.  Two feedback spokes compete
for one hub value and then reuse it in the opposite direction.}
\label{fig:main-h2}
\end{figure}
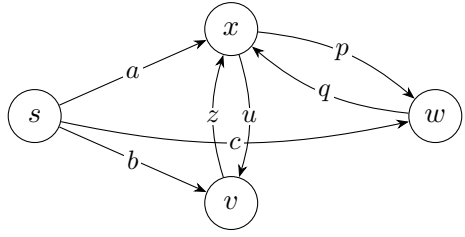

\begin{theorem}[Exact cyclic nonrectangularity]
\label{thm:main-h2}
\[
 \boxed{C^*(H_2)=4,\qquad P^*(H_2)=2,\qquad C_2^*(H_2)=5.}
\]
Hence the coordinatewise corner $(4,2)$ is infeasible.
\end{theorem}

The two scalar upper bounds are short explicit programs.  Four comparisons
follow the displayed distance formulas using four additions; a case split on
the three raw source values gives an addition-optimal two-addition program.
The addition lower bound is \cref{thm:main-addition}.  A geometric wall
argument gives $C^*\ge4$.

The final inequality $C_2^*\ge5$ is computer-assisted.  Its mathematical
reduction restricts a hypothetical four-comparison, two-addition program to
eight full-dimensional strict cells.  Recession removes literals without
changing any strict outcome.  A deterministic exhaustive closure then
enumerates every comparison of held registers and every legal addition,
including repeated operands and arbitrary mixed supports; free aliasing and
selection do not create coefficient forms.  The audited run visits 975
states and rules out all terminal states within the budget.  The exact cell
witnesses, certificate boundary, and independent verifier are given in
\cref{sec:cyclic-atom} and the repository verifier package.

The human part of the comparison lower bound is small enough to expose.  The
eight cells carry the affine output triples
\[
\begin{array}{llll}
 (a,b,c),&(a,b,a+p),&(a,a+u,c),&(a,a+u,a+p),\\
 (b+z,b,c),&(b+z,b,b+z+p),&(c+q,b,c),&(c+q,c+q+u,c).
\end{array}
\]
A three-comparison tree would have at most eight strict leaves, hence exactly
one leaf per cell and no query hyperplane cutting a cell interior.  The facet
adjacency graph is connected, but every separating facet normal cuts the
interior of another cell, a contradiction.  Thus $C^*\ge4$ without
computation.  Exhaustive closure is used only to show that the added
restriction $P\le2$ raises the comparison depth from four to five.

The verifier state consists of a held set of nonnegative integer coefficient
vectors, the surviving portions of the eight strict cones, and the remaining
budgets.  Addition inserts every componentwise sum of an unordered pair of
held vectors, including a repeated operand.  Comparison branches on every
difference of held vectors; exact integer cone certificates discard
infeasible outcomes.  A terminal state requires one fixed held output triple
on every surviving cone.  Aliases are therefore irrelevant, and no held
vector is assumed to encode a path.

\subsection{The shared-hub family}

For $k\ge1$, the graph $H_k$ has vertices $s,x,v_1,\ldots,v_k$ and arcs
\[
 s\to x:a,\qquad s\to v_i:b_i,\qquad
 v_i\to x:z_i,\qquad x\to v_i:u_i.
\]
Thus
\[
 d_x=\min\{a,b_1+z_1,\ldots,b_k+z_k\},\qquad
 d_i=\min\{b_i,d_x+u_i\}.
\]

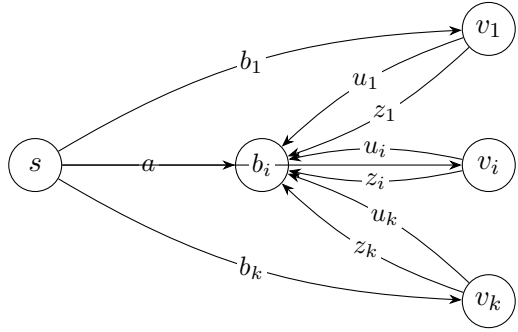
\begin{figure}[htbp]
\centering
\begin{tikzpicture}[
  vertex/.style={circle,draw,thin,minimum size=7mm,inner sep=1pt},
  every edge/.style={draw,thin,->,>=Stealth},
  every edge quotes/.style={font=\small,fill=white,inner sep=1pt}]
  \node[vertex] (s) at (0,0) {$s$};
  \node[vertex] (x) at (3,0) {$x$};
  \node[vertex] (v1) at (6,1.8) {$v_1$};
  \node[vertex] (vi) at (6,0) {$v_i$};
  \node[vertex] (vk) at (6,-1.8) {$v_k$};
  \path (s) edge["$a$"] (x)
        (s) edge[bend left=14,"$b_1$"'] (v1)
        (s) edge["$b_i$"'] (vi)
        (s) edge[bend right=14,"$b_k$"] (vk)
        (v1) edge[bend left=12,"$z_1$"] (x)
             edge[bend right=12,"$u_1$"'] (x)
        (vi) edge[bend left=12,"$z_i$"'] (x)
             edge[bend right=12,"$u_i$"] (x)
        (vk) edge[bend left=12,"$z_k$"'] (x)
             edge[bend right=12,"$u_k$"] (x);
\end{tikzpicture}
\caption{The family $H_k$.  Each spoke can contribute an inward hub
candidate or consume the final hub value through an outward candidate.}
\label{fig:main-hk}
\end{figure}

Every rooted order chooses exactly one forward class from each bidirected
spoke, so $P^*(H_k)=k$.  The unrestricted comparison optimum is linear: fix
the hub on an open product family to get $k$ independent output choices, and
use the direct formulas for the matching upper bound.

\begin{theorem}[Addition-optimal endogenous sorting]
\label{thm:main-sorting}
If $S(k)$ is the optimal pairwise-comparison sorting number, then
\[
 \boxed{\ceil{\log_2(k!)}\le C_k^*(H_k)\le S(k)+2k.}
\]
Consequently $C_k^*(H_k)=k\log_2k+O(k)$, whereas $C^*(H_k)=\Theta(k)$.
\end{theorem}

The lower bound is deliberately information-theoretic.  Addition optimality
forces the comparison transcript to distinguish all $k!$ scale-separated
spoke orders; aggregate comparisons are not asserted to simulate pairwise
key comparisons, so the stronger lower bound $S(k)$ is not claimed.

\section{The comparison--addition law on \texorpdfstring{$H_k$}{Hk}}
\label{sec:main-pareto}

For $0\le r\le k$, put
\[
 F_k(r)=C_{k+r}^*(H_k),\qquad
 \Lambda_{k,r}=\log_2\frac{k!}{r!(r+1)^{k-r}}.
 \label{eq:main-lambda}
\]
For $h\le k$, write $k=qh+t$, $0\le t<h$, and set
\(
 \Gamma_{k,h}=\log_2\bigl(k!/((h!)^q t!)\bigr)
\); set $\Gamma_{k,k+1}=0$.

\begin{theorem}[Deterministic entropy-tight Pareto law]
\label{thm:main-pareto}
Uniformly for $0\le r\le k$,
\[
 \boxed{
 \max\{k,\ceil{\Lambda_{k,r}}\}
 \le F_k(r)
 \le \Gamma_{k,r+1}+o(\Gamma_{k,r+1})+O(k).}
\]
In particular,
\[
 \boxed{F_k(r)=\Theta(k+\Lambda_{k,r})}.
\]
Whenever $\Lambda_{k,r}/k\to\infty$, the sharper conclusion is
$F_k(r)=(1+o(1))\Lambda_{k,r}$.  No additive-$O(k)$ deterministic
approximation to $\Lambda_{k,r}$ is claimed.
\end{theorem}

\paragraph{Lower bound: exposure, deadlines, entropy.}
Call the pair $Q_i=\{b_i,z_i\}$ exposed when one held support first contains
both primitives.  After $t$ additions, at most $t$ disjoint pairs are
exposed: add at most one edge per addition to a graph connecting the two
operand-support components; a component containing $p$ designated pairs
needs at least $p$ such edges.  This cumulative statement remains valid even
when one late addition exposes several pairs.

Fix a scale-separated main order $\pi$.  If $Q_{\pi_s}$ is hidden after
$s+r$ additions, a synchronized threat makes it essential to the hub and the
$k-s$ later spoke outputs.  These $k-s+1$ distinct values are all still
unmaterialized, but only $k-s$ additions remain---a contradiction.  Thus
$c_{\pi_s}\le s+r$.  On one fixed transcript, rank pairs by exposure time;
cumulative capacity gives $\operatorname{rank}(i)\le c_i$.  Compatible
permutations therefore satisfy $\sigma_s\le s+r$, of which there are exactly
\[
 B(k,r)=r!(r+1)^{k-r}.
\]
The $k!$ main orders and a binary transcript tree yield the entropy lower
bound in \cref{thm:main-pareto}.  This causal fan-out proof replaces an
invalid distinct-first-witness charging argument.

\paragraph{Upper bound: weak order, then roll.}
Use deterministic multiple selection to partition the $b_i$ into increasing
blocks of size at most $h=r+1$.  Its comparison cost is
$\Gamma_{k,h}+o(\Gamma_{k,h})+O(k)$.  Scan blocks until one crosses the
current hub candidate.  Complete earlier blocks need only inward candidates,
later blocks need only outward candidates, and at most $h-1=r$ elements of
the crossing block need both.  The addition count is at most $k+r$ and the
scan costs $O(k)$ comparisons.  Full scale synchronization, counting, and
rolling proofs appear in
\cref{sec:pareto-lower,sec:bbwo,sec:rolling,sec:pareto-law}.

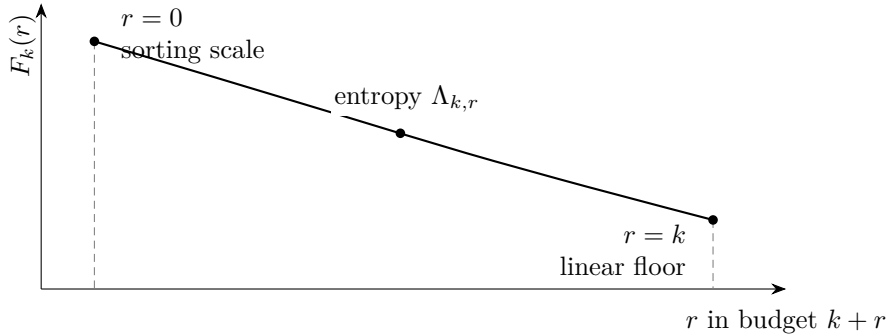
\begin{figure}[htbp]
\centering
\begin{tikzpicture}[x=0.88cm,y=0.63cm,font=\small]
  \draw[-{Stealth[length=2mm]},thin] (0,0) -- (11.2,0)
    node[below=4pt] {$r$ in budget $k+r$};
  \draw[-{Stealth[length=2mm]},thin] (0,0) -- (0,6)
    node[left=7pt,rotate=90] {$F_k(r)$};
  \draw[black,thick] plot[smooth] coordinates
    {(0.8,5.2) (2.6,4.45) (4.5,3.65) (6.4,2.85) (8.2,2.15) (10.1,1.45)};
  \draw[densely dashed,gray] (0.8,5.2)--(0.8,0);
  \draw[densely dashed,gray] (10.1,1.45)--(10.1,0);
  \fill (0.8,5.2) circle (1.8pt);
  \fill (5.4,3.27) circle (1.8pt);
  \fill (10.1,1.45) circle (1.8pt);
  \node[anchor=west,align=left] at (1.05,5.3)
    {$r=0$\\sorting scale};
  \node[align=center,fill=white,inner sep=1.5pt] at (5.5,4.0)
    {entropy $\Lambda_{k,r}$};
  \node[anchor=east,align=right] at (9.85,0.83)
    {$r=k$\\linear floor};
\end{tikzpicture}
\caption{The proven asymptotic regimes of the $H_k$ tradeoff.  The curve is
schematic; the theorem specifies the deterministic lower-order term.}
\label{fig:main-pareto}
\end{figure}

%% file: scalar_benchmark.tex
\section{From split resources to a same-program benchmark}
\label{sec:scalar-benchmark}

The preceding results expose a distinction that is invisible on DAGs.  There,
one program attains the lower-left corner of the resource rectangle.  On
\(H_2\), however, the best comparison program and the best addition program
cannot be the same program, and on \(H_k\) enforcing the addition optimum can
raise comparison cost from linear to sorting scale.  Hence a single-number
benchmark for the whole numerical computation must charge both resources
along one execution, as in \cref{eq:opt-dist}.

For comparison with the split frontier, define the balanced envelope
\[
 \mathsf{Env}(G,s)=\min_{b\ge0}\bigl(b+C_b^*(G,s)\bigr).
 \tag{66}\label{eq:balanced-envelope}
\]
This quantity chooses one feasible budgeted program; it is not
\(C^*+P^*\).

\begin{proposition}[The split frontier approximates the scalar optimum]
\label{prop:scalar-envelope}
For every finite reachable rooted topology,
\[
 \boxed{\OPT(G,s)\le \mathsf{Env}(G,s)\le2\OPT(G,s).}
 \tag{67}\label{eq:scalar-envelope}
\]
\end{proposition}
\begin{proof}
For any budget \(b\), one program attaining \(C_b^*\) has at most \(b\)
additions and at most \(C_b^*\) comparisons on every input, hence scalar cost
at most their sum.  Minimizing gives the first inequality.

Conversely, let \(A\) attain \(K=\OPT(G,s)\), and put
\(b=\sup_w P_A(w)\), \(c=\sup_w C_A(w)\).  Both are at most \(K\), and the
same program witnesses \(C_b^*(G,s)\le c\).  Therefore
\(
 \mathsf{Env}(G,s)\le b+C_b^*(G,s)\le b+c\le2K.
\)
No step combines coordinate optima belonging to different programs.
\end{proof}

The factor-two relation is useful but not the endpoint.  The next section
gives an exact game for \(\OPT\) itself.  The game charges an addition and a
comparison identically and takes the maximum over feasible comparison outcomes
on the same path, so its value is the scalar optimum rather than the sum of
two coordinatewise maxima.

%% file: universal_main.tex
\section{Exact numerical universal optimality}
\label{sec:main-universal}

The benchmark $\OPT$ allows a different finite program for each topology.
Exact universal optimality asks whether one interpreter can match those
nonuniform programs without paying an extra charged operation.

On a fixed branch, every register is an affine form
$\langle a,w\rangle+\alpha$.  A recession argument scales the input while
holding the finite literal set fixed.  It converts every finite execution to
an equivalent homogeneous execution on the nonnegative orthant: literals
become the zero form, strict comparisons keep their eventual sign, and
equality branches are retained exactly when their homogeneous normals vanish.
The normal form is used only after a finite program and branch are fixed.

In the homogeneous game, $M$ is the finite set of coefficient vectors of
held registers and $\Gamma$ is the sequence of physical comparisons already
made.  They define the transcript region (a relatively open cone with its
equality faces retained explicitly)
\[
 \mathcal P(\Gamma)=\{w\in\mathbb R_{\ge0}^E:
                 w\text{ satisfies every outcome in }\Gamma\}.
\]
A label $v$ is \emph{ready} only if one fixed vector in $M$ equals $d_w(v)$
for every $w$ in this entire cone.  A move either adds two vectors in $M$ or
compares two of them and follows every feasible outcome.  Let
$\WIN_G(M,\Gamma,K)$ mean that all continuations can be made terminal within
$K$ further moves.

\begin{theorem}[Transcript-cone equivalence]
\label{thm:main-game}
For every topology and integer $K\ge0$,
\[
 \boxed{
 \WIN_G(\{0,e_1,\ldots,e_m\},\varnothing,K)
 \quad\Longleftrightarrow\quad \OPT(G,s)\le K.}
\]
\end{theorem}

A program induces a game strategy by recording its actually materialized
forms and physical comparison outcomes.  Correctness at a leaf supplies one
fixed ready register per label over the whole cone; otherwise two inputs with
the same transcript would require different materialized outputs.  Conversely,
every game move is a legal physical operation, and readiness supplies legal
output registers.  Recession transfers the two directions between affine and
homogeneous executions, including equality faces, zero weights, parallel
arcs, and zero-weight cycles.  Duplicate forms may be discarded because
recreating one cannot help.

The quantifier over the whole cone is essential.  It prevents a leaf from
selecting one register at one weighting and a different register at another
weighting that produced the same physical transcript.  Conversely, it does
not require one common shortest path: different paths may certify the same
fixed held form at different points of the cone.  Zero-weight cycles add no
output form because every distance has a simple shortest-path representative.
Parallel primitives remain separate initial coordinates.  When a transcript
has no strict outcome, its cone also contains the all-zero input; transcripts
with strict outcomes are handled directly on their relatively open regions,
not by a generic-limit argument.

\begin{corollary}[No numerical uniformity tax]
\label{cor:main-uniform}
There is one deterministic uniform interpreter $U_{\rm pspace}$ such that
\[
 \boxed{T_{U_{\rm pspace}}(G,s)=\OPT(G,s)}
\]
on every finite reachable topology.
\end{corollary}

The interpreter computes the least winning budget and then chooses the first
winning legal move under a canonical ordering.  This is genuinely uniform:
the topology-specific strategy is reconstructed from the finite input graph,
not supplied as advice.

\begin{theorem}[Complexity of policy generation]
\label{thm:main-pspace}
An optimal next action in the transcript-cone game can be generated in
polynomial space.  The direct search may take
$2^{\operatorname{poly}(n+m)}$ ordinary time.
\end{theorem}

Bellman--Ford gives a polynomial charged-budget cap.  Within that cap,
coefficient entries and the held-form set have polynomial encodings, and
there are polynomially many candidate moves at a state.  Transcript
feasibility is linear programming.  Nonreadiness has a polynomial certificate
consisting, for each held candidate, of a cone point and a strictly better
simple-path form, so readiness lies in coNP.  Depth-first minimax stores only
one polynomial-size state per recursion level.  These claims concern space,
not time; the detailed normal-form, readiness, and search proofs are in
\cref{sec:universal-numerical,app:recession,app:ready-pspace}.

%% file: efficient_main.tex
\section{Efficient uniformity}
\label{sec:main-efficient}

The ratio-one interpreter settles existence but not efficient navigation.
A topology-only reduction relates the best current polynomial-time bound to
the same benchmark.

Delete every arc $(u,v)$ for which $v$ dominates $u$; such an arc only closes
a nonnegative subwalk.  Reduce each remaining parallel source class to its
minimum, at total comparison cost
$\sigma=\sum_v(p_v-1)_+$.  Let $E_\circ$ be the remaining primitives with
nonsource tail and $k=|E_\circ|$.  The canonical core $H$ consists of their
endpoints, the primitives of $E_\circ$, and one virtual source register for
each retained source class.  Then $|V(H)|\le2k+1$ and $|E(H)|\le3k$.

\begin{theorem}[Active-core factorization]
\label{thm:main-core}
\[
 \boxed{
 \OPT(H,q)\le\OPT(G,s)
 \le\sigma+\OPT(H,q)
 \le2\OPT(G,s).}
\]
Moreover, $\sigma+k\le2\OPT(G,s)$.
\end{theorem}

The simulations merely set deleted inputs to zero or alias source-class
minima, so they preserve arbitrary mixed-sum competitors.  The quantitative
charge is same-input: two opposite low--high dominator orders ensure that one
rooted order exposes enough active edges, while the potential-tie construction
forces the corresponding comparisons and additions on a single execution.

More precisely, for a rooted order $\prec$, let $m_\prec$ and $q_\prec$ be
the forward primitive-arc and endpoint-class counts.  The potential-tie input
forces at least $m_\prec-n+1+q_\prec-d_s$ charged operations on one
execution.  A low--high preorder of the dominator tree and the preorder
obtained by reversing every child list place each active nonsource arc forward
in at least one of the two orders, except for the explicitly charged source
multiplicities.  Averaging gives $\sigma+k\le2\OPT$.  This is stronger than
adding separate worst-case comparison and addition lower bounds.

Using the deterministic directed-SSSP algorithm of Duan, Mao, Shu, and
Yin~\cite{duan2026faster} on this core gives the current efficient endpoint.

\begin{theorem}[Efficient topology-wise bound]
\label{thm:main-efficient}
There is one explicit deterministic uniform exact-\DIST algorithm with
$O(m+n)$ topology preprocessing and
\[
 \boxed{
 T_U(G,s)=O\!\left(
  \OPT(G,s)\sqrt{\log_+(2+\OPT(G,s))
                       \log\log_+(4+\OPT(G,s))}
 \right).}
\]
Its ordinary query work is $O(m+n+C_U+P_U)$ under the stated bookkeeping
convention.
\end{theorem}

The external SSSP routine is used only as a black box; the new statement is
the active-core lifting to a topology-sensitive numerical benchmark.  Thus
exact optimal policy existence is settled, while a polynomial-time
constant-competitive navigation rule remains open.  Full deletion,
virtual-source, dominator-order, and accounting arguments are in
\cref{sec:efficient-universality,app:active-core}.

%% file: related_work.tex
\section{Related work}
\label{sec:related}

\paragraph{Universal optimality and output contracts.}
Haeupler, Hlad\'{i}k, Rozho\v{n}, Tarjan, and T\v{e}tek prove that Dijkstra,
with a suitable beyond-worst-case heap, is universally optimal for producing
vertices in distance order~\cite{haeupler2024universal}.  Van der Hoog,
Rotenberg, and Rutschmann give a simpler construction and
analysis~\cite{vanderhoog2025simpler}.  Their benchmark distinguishes distance
orders as outputs.  Our topology-specific benchmark instead asks only for
materialized labeled distances; \cref{app:dist-do-separation} makes the
boundary explicit.

\paragraph{Directed numerical SSSP below sorting.}
Duan et al. broke the sorting barrier for deterministic directed SSSP in the
comparison--addition setting~\cite{duan2025breaking}.  The subsequent bound of
Duan, Mao, Shu, and Yin~\cite{duan2026faster} is the efficient black box in
\cref{sec:efficient-universality}.  These works give uniform worst-case
running-time algorithms; they do not characterize the fixed-topology
same-program optimum \(\OPT\).  Our lifting theorem is not a new construction
of their collective shortest-path routine.

\paragraph{Algebraic path complexity.}
Mahr gives graph-sensitive semiring-operation lower bounds for fixed-network
path problems and tight adapted schemes on cycle-free graphs
\cite{mahr1982}.  Jerrum and Snir study exact straight-line semiring
complexity~\cite{jerrumSnir1982}, and Jukna develops tropical-circuit and
dynamic-program lower bounds~\cite{jukna2015}.  These are important
predecessors, but their program/output requirements differ from adaptive
comparison--addition programs computing only explicit single-source numerical
coordinates.  Our arbitrary-topology addition proof works directly with one
adaptive branch and does not convert it to a path circuit.

\paragraph{Comparison--addition and relaxation models.}
Spira and Pan study verification lower bounds for shortest paths
\cite{spiraPan1975}; Pettie studies comparison--addition APSP
\cite{pettie2002}; and Pettie and Ramachandran study real-weighted undirected
SSSP~\cite{pettieRamachandran2005}.  Eppstein analyzes nonadaptive relaxation
sequences~\cite{eppstein2023relaxation}, and Atalig et al. give lower bounds
for an adaptive relaxation-based model~\cite{atalig2024adaptive}.  Our
competitors may compare arbitrary held aggregates and use mixed sums,
equality, literals, and cross-output sharing; conversely, our conclusions are
specific to fixed directed topologies and labeled \DIST.

\paragraph{Selection and partial orders.}
The upper side of the \(H_k\) law uses deterministic multiple selection
\cite{kaligosi2005multiple} and the related partial-order-production framework
\cite{cardinal2010partial}.  Their deterministic guarantees have the form
information lower bound plus a lower-order information term and $O(k)$; we do
not strengthen that to a uniform additive-$O(k)$ statement.  No novelty is
claimed for these primitives.  The contribution is the SSSP reduction showing
that an addition budget induces banded-permutation entropy and that a
bounded-width weak order supports the rolling computation.

\paragraph{Structural decompositions.}
Tarjan's dominator strong-component framework decomposes algebraic path
problems~\cite{tarjan1981path}.  Linear-time dominator algorithms support our
topology-only active reduction~\cite{alstrup1999dominators}.  Recent
Acyclic-Connected-tree work gives topology-dependent algorithms for SPT and
DIST under a different decomposition objective~\cite{stefansson2025acyclic}.
Our canonical core is used only to preserve and lift the unrestricted scalar
benchmark.

\begin{table}[htbp]
\centering
\scriptsize
\begin{tabularx}{\textwidth}{@{}>{\raggedright\arraybackslash}p{0.21\textwidth}>{\raggedright\arraybackslash}p{0.28\textwidth}Y@{}}
\toprule
Closest line of work & Its benchmark or restriction & Boundary of the present claims \\
\midrule
Haeupler et al.; van der Hoog et al.
& topology-wise universal optimality for producing a distance order
& labeled \DIST only; exact charged numerical optimum rather than an
order-sensitive comparison benchmark \\
Duan et al. 2025; Duan et al. 2026
& uniform worst-case running time for directed numerical SSSP
& fixed-topology resource identities and lifting to $\OPT$; no novelty claim
for their SSSP routines \\
Mahr; Jerrum--Snir; Jukna
& straight-line semiring or tropical circuit complexity
& adaptive comparisons, affine literals, mixed held sums, and explicit
labeled register output \\
Pettie; Pettie--Ramachandran
& comparison--addition upper bounds for APSP or undirected SSSP
& directed fixed-topology lower bounds and the two-resource feasible region \\
Eppstein; Atalig et al.
& nonadaptive or adaptive relaxation-based lower bounds
& competitors may compare arbitrary held aggregates and are not restricted to
relaxations \\
Multiple selection and partial-order production
& information-efficient production of selected ranks or a target partial order
& standard upper-bound primitive; the addition-budget reduction and causal
fan-out lower bound are graph-specific \\
\bottomrule
\end{tabularx}
\caption{Scope comparison with the closest prior frameworks.}
\label{tab:prior-scope}
\end{table}

Our novelty statements are correspondingly narrow.  We claim the displayed
fixed-topology identities and separations in the stated \CA model, not priority
for comparison--addition shortest paths, graph-sensitive algebraic bounds,
selection, or universal optimality as general ideas.  A targeted literature
search through September 2026 located no source stating the same
$\rhoF$ identity, cyclic nonrectangle, arithmetic-budget entropy law, or
transcript-cone equivalence; this negative search is evidence for positioning,
not a proof of priority.

%% file: discussion.tex
\section{Interpretation of the resource theory}
\label{sec:discussion}

Three distinctions organize the results.  The first is between output
contracts.  A comparison lower bound for producing a distance order need not
apply to labeled \DIST, because the latter may materialize all coordinates
without certifying their relative order.  Conversely, the shared-hub family
shows that removing the order from the output does not remove ordering from
every optimal computation: a tight arithmetic budget can force a transcript
with sorting-scale information.  The relevant order is therefore endogenous
to the resource constraint, not prescribed by the task.

The second distinction is between coordinatewise and same-program
optimization.  On a DAG, the potential-tie input makes the comparison and
addition lower bounds coexist and a topological program attains their common
corner.  The graph $H_2$ is the smallest audited obstruction to extrapolating
that geometry through cycles.  Once the feasible region is nonrectangular,
$C^*+P^*$ ceases to describe any physical algorithm.  The benchmark $\OPT$
repairs the quantifiers: it adds the two charges on an execution, takes the
worst weighting for that program, and only then minimizes over programs.  The
factor-two envelope in \cref{prop:scalar-envelope} is useful for intuition,
but the transcript-cone game is needed for an exact statement.

The third distinction is between existence and navigation.  The ratio-one
interpreter is not merely a nonuniform optimal policy restated.  Its canonical
minimax search reconstructs the next topology-specific action from the graph,
the held forms, and the observed cone.  This removes every charged-operation
advantage of nonuniform advice.  It does not compress the search needed to
find that action.  The active-core theorem addresses the latter problem by
shrinking the numerical instance to size $O(\OPT)$ and invoking a current
uniform SSSP routine, at the price of a sublogarithmic factor.  A future
constant-competitive algorithm would have to navigate enough of the exact
game implicitly, or find a different certificate of a good move.

The proof methods follow the same hierarchy.  The addition theorem is
algebraic and arbitrary-topology: kernel perturbations and quotient fibers
exclude mixed-sum shortcuts.  The $H_k$ lower bound is causal: support
connectivity limits cumulative exposure, while fan-out turns a hidden pair
into more unmaterialized labeled outputs than the remaining addition budget
can create.  The universal theorem is semantic: readiness is quantified over
an entire transcript cone rather than over one representative weighting.
These mechanisms are complementary; none is presented as a generic lower
bound for every computational model of shortest paths.

Several natural extensions remain outside the present scope.  Randomized
programs would require distributional versions of both the transcript-cone
game and the $H_k$ threat family.  Allowing subtraction or scalar
multiplication destroys the monotone support arguments used for additions.
Replacing independent arc weights by correlated or succinctly encoded inputs
changes both the coefficient geometry and the appropriate cost measure.  Even
within the present model, an exact comparison characterization for arbitrary
cyclic topologies is open.  The current paper therefore supplies a complete
addition coordinate, exact and asymptotic cyclic witnesses, and an exact
scalar game---not a closed formula for every two-resource frontier.

%% file: limitations.tex
\section{Limitations and open problems}
\label{sec:limitations}

\paragraph{Model scope.}
The results concern deterministic finite comparison--addition programs on
independent nonnegative real primitive weights.  They do not provide
randomized lower bounds, bit-complexity bounds for encoded numbers, or
arithmetic results in a model with free subtraction or scalar multiplication.
The requirement that each labeled distance be an actually held register is
essential.

\paragraph{Scope of the resource characterizations.}
The addition coordinate is exact on every topology and the full region is
exact on DAGs.  The entropy-tight Pareto law is a theorem for the shared-hub
family \(H_k\), not an arbitrary-topology frontier formula.  Deterministically
we prove a lower-order entropy overhead plus $O(k)$, not a uniform
additive-$O(k)$ correction to $\Lambda_{k,r}$.  Exact comparison complexity on
general cyclic topologies remains open.

\paragraph{Computer-assisted boundary.}
The endpoint assertion \(C_2^*(H_2)=5\) uses an exact finite exhaustive
certificate after a mathematical reduction to homogeneous coefficient states.
The verifier checks every legal addition and comparison action and validates
cone feasibility or infeasibility by integer certificates; see
\cref{app:h2-obstruction}.  The other principal theorems have symbolic proofs
independent of finite regression.

\paragraph{Numerical versus efficient universal optimality.}
The transcript-cone interpreter proves ratio one only for charged additions
and comparisons.  Its exact action search may take exponential ordinary time,
and the PSPACE theorem is not a claim of practicality.  Conversely, the
active-core algorithm is efficient but retains a sublogarithmic competitive
factor.

\begin{openproblem}[Efficient exact uniformity]
\label{open:efficient-uniformity}
Is there an explicit uniform deterministic algorithm that, on every rooted
directed topology, computes exact labeled \DIST with
\[
 \operatorname{RAM}_U(G,s,w)=O\!\left(m+n+\OPT(G,s)\right)
\]
for all nonnegative weightings, and hence uses \(O(\OPT(G,s))\) charged
numerical operations?
\end{openproblem}

The ratio-one theorem changes the interpretation of this question.  The
remaining obstacle is not existence of an optimal numerical policy, nor the
size of an unfolded policy tree.  It is fast, implicit navigation of a policy
from the topology and the observed transcript.

%% file: conclusion.tex
\section{Conclusion}
\label{sec:conclusion}

Removing distance ordering from exact directed SSSP exposes a distinct
topology-specific theory.  Arithmetic materialization has the exact invariant
\(\rhoF\); DAGs align comparison and addition costs; and a minimal active
cycle breaks that alignment.  In the shared-hub family, saving additions
forces enough internal information to recover sorting scale, yielding an
entropy-tight Pareto law although the requested output remains a labeled
vector.

That conflict identifies the right total benchmark: the same-program
worst-case sum \(\OPT\).  Transcript cones then show that one fixed uniform
interpreter matches every specialized numerical optimum exactly, with a
polynomial representation of the online state.  Nonuniformity therefore
offers no advantage in charged numerical operations.  What it can still offer
is efficient access to the right next operation.  Canonical active-core
lifting and collective directed SSSP give the current sublogarithmic
competitive factor, leaving \cref{open:efficient-uniformity} as the central
question: can exact numerical universal optimality be navigated in linear
topology overhead and constant amortized ordinary work per charged action?

%% file: topology_parameter.tex
\section{The forward topology parameter}
\label{sec:topology}

For distinct vertices $u,v$, the endpoint class
\[
 E_{uv}=\{e\in E:\operatorname{tail}(e)=u,
                      \operatorname{head}(e)=v\}
\]
collects parallel arcs.  Let $\overline E$ be the set of nonempty endpoint
classes and let $q=|\overline E|$.  The total number of arcs is $m=|E|$.
For a class $E_{uv}$, its minimum input value can be selected with
$|E_{uv}|-1$ comparisons and no additions.  A source-out class is already a
raw input candidate; a non-source tail requires a path-extension addition when
its tail is settled before its head.

\begin{definition}[Rooted order]
\label{def:rooted-order}
A total order $\prec$ of $V$ is \emph{rooted} at $s$ if $s$ is first and every
$v\ne s$ has an endpoint class $(u,v)$ with $u\prec v$.
\end{definition}

For a rooted order define
\[
 f_G(\prec)=
 \bigl|\{(u,v)\in\overline E:u\ne s,\ u\prec v\}\bigr|,
 \qquad
 \rhoF(G,s)=\max_{\prec\text{ rooted}} f_G(\prec).
 \tag{16}\label{eq:rho}
\]
The maximum is finite.  On a simple loopless graph, the related arc-level
maximum-forward count differs by the source out-degree; endpoint classes are
the right form for multigraph addition accounting.

\begin{lemma}[DAG specialization]
\label{lem:dag-rho-new}
If $G$ is a reachable directed acyclic multigraph and $d_s$ is its number of
source-out endpoint classes, then
\[
 \rhoF(G,s)=q-d_s.
 \tag{17}\label{eq:dag-rho}
\]
\end{lemma}
\begin{proof}
A topological order is rooted and makes every non-source endpoint class
forward, giving $q-d_s$.  No order can count more than all those classes.
\end{proof}

\begin{lemma}[Realizing a rooted order]
\label{lem:realize-rooted-order-new}
Every rooted order $s=v_1\prec\cdots\prec v_n$ is the strict Dijkstra
settlement order of some strictly positive weighting.  On that weighting every
forward endpoint class is processed before its head is settled.
\end{lemma}
\begin{proof}
Choose $0=\pi(s)<\pi(v_2)<\cdots<\pi(v_n)<1$.  Give every forward arc the
weight $\pi(v)-\pi(u)$ and every backward arc a weight in $(2,3)$.  The
forward classes contain a rooted arborescence, so each vertex has a forward
path of length $\pi(v)$.  Any path using a backward arc has length greater
than two, while every forward path has length below one.  Hence
$d(v)=\pi(v)$ and the strict settlement order is the chosen one.
\end{proof}

\begin{remark}
Computing $\rhoF$ itself is NP-hard by the standard reduction from directed
feedback arc set~\cite{karp1972reducibility}: add a new source with an arc to
every old vertex, so every order is rooted and maximizing forward classes is
maximum acyclic subgraph.
The addition-optimal uniform algorithm below never computes $\rhoF$.
\end{remark}

%% file: addition_coordinate.tex
\section{The addition coordinate on arbitrary topologies}
\label{sec:additions}

We first settle the resource coordinate that is insensitive to the comparison
budget.  The proof is deliberately phrased in terms of coefficient records,
not path registers.  This distinction is what makes the result apply to
arbitrary adaptive programs with mixed and non-path sums.

\subsection{A uniform upper bound}

Reduce each endpoint class $E_{uv}$ to a minimum register
$\mu_{uv}=\min_{e\in E_{uv}}w_e$.  Run Dijkstra with a fixed deterministic tie
rule.  When $u$ is settled, ignore a class whose head is already settled.  If
$u=s$, the class minimum is already a primitive candidate.  Otherwise create
the single register $d(u)+\mu_{uv}$ and relax the head.  Copy the final
settled register to each labeled output.

\begin{lemma}[Endpoint-class Dijkstra bound]
\label{lem:addition-upper-new}
For every nonnegative weighting $w$,
\[
 P_A(G,s,w)\le \rhoF(G,s).
 \tag{18}\label{eq:addition-upper}
\]
For every fixed topology, equality holds for some strictly positive weighting.
\end{lemma}
\begin{proof}
The settlement order is rooted: every nonsource vertex is first reached from
the source or from a previously settled vertex.  The algorithm adds once for
each class $(u,v)$ with $u\ne s$ and $u$ earlier than $v$, and never for a
backward class or a class entering an already settled head.  The count is
therefore at most $f_G(\prec_w)\le\rhoF(G,s)$.  Zero weights and ties only
change the deterministic rooted order.  Conversely, choose a rooted order
maximizing $f_G$ and apply \cref{lem:realize-rooted-order-new}; the algorithm
performs one addition for every counted class.
\end{proof}

The algorithm is uniform.  Class formation can be done by topology
preprocessing, class reduction costs exactly $m-q$ comparisons, and the
remaining control work is ordinary Dijkstra bookkeeping.  The theorem below
concerns only additions; no computation of the NP-hard parameter $\rhoF$ is
needed.

\subsection{Three algebraic tools}

Fix a rooted order $\prec$ and let $H$ contain all arcs in endpoint classes
forward under $\prec$.  Then $H$ is a rooted spanning DAG.  Write $B_H$ for
its signed vertex--arc incidence matrix, with a column
$\mathbf 1_v-\mathbf 1_u$ for an arc $(u,v)$.

\begin{lemma}[Common-endpoint path differences]
\label{lem:path-kernel-new}
Let $H$ be a rooted reachable DAG.  If $\mathcal P_v$ is the set of directed
$s$--$v$ paths, then
\[
 \ker B_H=spanof\{\chi_P-\chi_Q:
          v\in V,\ P,Q\in\mathcal P_v\}.
 \tag{19}\label{eq:path-kernel}
\]
\end{lemma}
\begin{proof}
Every displayed difference has zero incidence.  Choose a rooted
out-arborescence $T$.  For each non-tree arc $e=(u,v)$, the path $T_u e$ is an
$s$--$v$ path: otherwise $T_u$ and $e$ would contain a directed cycle.  The
vector
\[
 z_e=\chi_{T_u}+\mathbf 1_e-\chi_{T_v}
\]
is therefore a common-endpoint path difference.  Its own non-tree coordinate
is one and all other non-tree coordinates are zero, so the $z_e$ are linearly
independent.  There are $|E(H)|-(|V|-1)=\dim\ker B_H$ of them, proving the
claim.  Parallel arcs are covered because they are simply distinct non-tree
coordinates when necessary.
\end{proof}

\begin{lemma}[Generic-tie and slack projection]
\label{lem:generic-slack-new}
Let $E(H)$ and $E(G)\setminus E(H)$ be denoted by $E_H$ and $E_O$.  For a
fixed finite program, there is a strictly positive weighting with the
following properties on the reached execution:

\begin{enumerate}[leftmargin=*,itemsep=2pt]
  \item every $H$-path to $v$ has the same length $\pi(v)$ and every path
        using an outside arc is strictly longer;
  \item if two reached registers are equal, their outside coefficient vectors
        agree, their $H$-coefficient difference is in $\ker B_H$, and their
        constant terms agree;
  \item a correct output for $v$ has no outside coefficient, has zero
        constant term, and its $H$ coefficient has incidence
        $\mathbf 1_v-\mathbf 1_s$; and
  \item if $z$ is orthogonal to all reached equality normals projected to
        $E_H$, then both sufficiently small perturbations
        $w_H\pm\varepsilon z$ preserve positivity, the outside-path gap, and
        the complete adaptive branch.
\end{enumerate}
\end{lemma}
\begin{proof}
Let $K$ be the countable field generated by the program's finite literals.
Choose $0=\pi(s)<\pi(v)<1$ algebraically independent over $K$ in the
chosen order cone.  For $e=(u,v)\in E_H$, set
$w_e=\pi(v)-\pi(u)$.  For each $e\in E_O$, choose an independent slack
$\sigma_e\in(2,3)$ over $K(\pi(v):v\ne s)$ and set
$w_e=\pi(v)-\pi(u)+\sigma_e$.  These weights are positive.  Telescoping
shows that an $H$-path to $v$ has length $\pi(v)<1$, whereas a path using an
outside arc has length greater than two.

On the finite branch, each register has the affine form
$a_H\cdot w_H+a_O\cdot w_O+\gamma$.  An equality between two such forms is a
polynomial identity in the independent slacks and potentials.  Independence
first forces $a_O=a'_O$, and then forces
$B_H(a_H-a'_H)=0$ and $\gamma=\gamma'$.  Applying the same argument to a
correct output, whose value is $\pi(v)$, forces $a_O=0$, $\gamma=0$, and the
displayed divergence.

For completeness, the perturbation radius can be chosen explicitly after the
finite branch is fixed.  Let $g>0$ be the minimum baseline gap between an
outside-arc simple path and an $H$-path to the same endpoint, and let $M$ be
the maximum absolute change of a path difference in direction $z$.  Take
$\varepsilon<g/(2\max\{1,M\})$ and also smaller than
$w_e/(2|z_e|)$ for every nonzero $z_e$.  For each strict comparison on the
branch, let $\mu_j$ be its positive baseline margin and $\theta_j$ its slope
in direction $z$; additionally take
\[
 \varepsilon<\min_j
 \frac{\mu_j}{2\max\{1,|\theta_j|\}}.
 \tag{20}\label{eq:branch-radius}
\]
Equality comparisons are unchanged because their projected normals are
orthogonal to $z$, and strict comparisons retain their signs.  Finiteness of
the path and branch sets makes all minima positive whenever the corresponding
constraint is nonvacuous.  This proves all four claims for both signs.
\end{proof}

\begin{lemma}[Fiber opening]
\label{lem:fiber-opening-new}
Let $L:\R^M\to\R^d$ be linear.  Suppose a coefficient-vector execution starts
with a finite set $S_0$ and inserts one vector per charged binary addition.
For a finite set $S$, write
\[
 D_L(S)=\spanof\{a-b:a,b\in S,\ La=Lb\}.
 \tag{21}\label{eq:fiber-difference}
\]
If the required equality normals span an $R$-dimensional subspace of
$D_L(S)$, and the required outputs occupy at least $k$ $L$-fibers absent from
$S_0$, then an execution with $P$ additions satisfies
\[
 P\ge R-\dim D_L(S_0)+k.
 \tag{22}\label{eq:fiber-opening}
\]
\end{lemma}
\begin{proof}
Insert the vectors in execution order.  An insertion into an already occupied
fiber can increase the within-fiber difference space by at most one.  The
first insertion in a new fiber increases it by zero, because it has no mate
in that fiber.  If $h$ new fibers are opened, then
\[
 \dim D_L(S)\le \dim D_L(S_0)+P-h.
\]
The hypotheses give $\dim D_L(S)\ge R$ and $h\ge k$, which rearranges to
the claimed inequality.  The argument counts every inserted vector, including
ones in registers later discarded.
\end{proof}

\begin{corollary}[Shortest-path fiber bound]
\label{cor:sp-fiber-new}
In the setting of \cref{lem:generic-slack-new}, suppose the projected equality
normals span $\ker B_H$.  If $m_H$ is the number of arcs of $H$, $q_H$ its
number of endpoint classes, and $d_s$ its source-out class count, then the
execution uses at least
\[
 (m_H-n+1)-(m_H-q_H)+(n-1-d_s)=q_H-d_s
 \tag{23}\label{eq:projected-addition-lower}
\]
additions.
\end{corollary}
\begin{proof}
Project every coefficient vector to $\R^{E_H}$.  The free initial vectors
are zero and the units of the $H$ arcs.  Units in one endpoint class have the
same divergence, so
\[
 \dim D_{B_H}(S_0)=\sum_{E_{uv}\subseteq E_H}(|E_{uv}|-1)=m_H-q_H.
 \]
The incidence rank is $n-1$, hence
$\dim\ker B_H=m_H-n+1$.  A nonsource vertex without a source-out class has
output divergence $\mathbf 1_v-\mathbf 1_s$ in a fiber absent from the free
initial set.  There are $n-1-d_s$ such distinct fibers.  Apply
\cref{lem:fiber-opening-new} with
$R=m_H-n+1$ and $k=n-1-d_s$.

This projection is essential.  Outside primitive coordinates become the one
zero vector rather than separate occupants of a zero-divergence fiber.  Thus
the argument permits outside-contaminated mixed sums and does not assume a
path-register normal form.
\end{proof}

\subsection{The exact theorem}

\begin{theorem}[Arbitrary-topology addition theorem]
\label{thm:arbitrary-addition-new}
For every finite reachable rooted directed multigraph in the model of
\cref{sec:model},
\[
 \boxed{P^*(G,s)=\rhoF(G,s).}
 \tag{24}\label{eq:arbitrary-addition-theorem}
\]
The statement remains true with unlimited additional comparisons and is
unchanged by self-loops.
\end{theorem}
\begin{proof}
The upper bound is \cref{lem:addition-upper-new}.  For the lower bound, fix
any correct adaptive program and any rooted order.  On its forward DAG $H$,
use the algebraically independent potential/slack weighting of
\cref{lem:generic-slack-new}.  If the projected equality normals failed to
span $\ker B_H$, choose a nonzero direction $z$ in the orthogonal complement.
By \cref{lem:path-kernel-new}, some common-endpoint path pair has different
$z$-slopes.  The two perturbations from
\cref{lem:generic-slack-new} follow the same complete branch but require
different minimum affine outputs at that endpoint, a contradiction.  Hence
the equality normals span the full kernel, and
\cref{cor:sp-fiber-new} gives at least $q_H-d_s=f_G(\prec)$ additions on
this program-dependent positive input.  Maximize over rooted orders.  The
self-loop claim follows by ignoring loops in one direction and substituting
zero loop weights in the other; nonnegative loops cannot improve a distance.
\end{proof}

\begin{corollary}[Uniform addition optimality]
\label{cor:uniform-addition-new}
Endpoint-class Dijkstra is one uniform algorithm satisfying, for every fixed
topology,
\[
 \sup_{w\ge0}P_A(G,s,w)=P^*(G,s)=\rhoF(G,s).
\]
\end{corollary}
\begin{proof}
Combine \cref{lem:addition-upper-new} with
\cref{thm:arbitrary-addition-new}.
\end{proof}

%% file: dag_frontier.tex
\section{The exact frontier on directed acyclic graphs}
\label{sec:dag-frontier}

The arbitrary-topology addition theorem does not by itself determine a joint
frontier.  On a DAG, however, all shortest-path candidates can be exposed in
one topological schedule, and the two resources align exactly.

Let $G$ be a reachable directed acyclic multigraph.  Its $n$ vertices and $m$
arcs have $q$ nonempty endpoint classes, of which $d_s$ leave the source.
Every endpoint class enters a nonsource vertex, since an arc entering $s$
would create a directed cycle with a source path to its tail.

\subsection{A simultaneous upper bound}

First reduce each endpoint class to its minimum primitive value.  Process
vertices in a fixed topological order.  A source-out class contributes its
primitive minimum directly.  For every other class $(u,v)$, materialize
$d(u)+\mu_{uv}$ once.  At $v$, compare the incoming class candidates and
retain the first minimum under a fixed equality rule.

The class reductions cost
\[
 \sum_{E_{uv}}(|E_{uv}|-1)=m-q
 \tag{25}\label{eq:dag-class-comparisons}
\]
comparisons.  After reduction, a nonsource vertex with $r_v$ incoming
classes needs $r_v-1$ comparisons, and
\[
 \sum_{v\ne s}(r_v-1)=q-(n-1).
 \tag{26}\label{eq:dag-vertex-comparisons}
\]
The total is $m-n+1$.  The number of additions is exactly the number of
non-source classes, $q-d_s$.  All choices remain valid for zero weights and
ties.

\subsection{Why the lower bounds share an input}

For completeness, we spell out the comparison part of the generic-tie
argument.  Choose algebraically independent potentials
\[
 0=\pi(s)<\pi(v)<1
 \tag{27}\label{eq:dag-potentials}
\]
and set $w_{uv}=\pi(v)-\pi(u)$ on every arc.  Every source path to $v$ has
length $\pi(v)$.  On the finite execution of an arbitrary correct adaptive
program, every register is an affine function with a nonnegative integral
coefficient vector.  Algebraic independence implies that each reached
equality comparison has coefficient normal in $\ker B_G$, and that a correct
output for $v$ has divergence $\mathbf 1_v-\mathbf 1_s$ and zero constant.

Let $Z$ be the span of the equality normals on this execution.  If
$Z\subsetneq\ker B_G$, choose $0\ne z\in\ker B_G\cap Z^\perp$.  By
\cref{lem:path-kernel-new}, there are two source paths $P,Q$ to one endpoint
with different $z$-slopes.  For both sufficiently small signs of
$w+\varepsilon z$, all strict comparisons retain their outcome and all
equality comparisons remain equal.  Thus the same affine output register is
used on both sides.  At the displayed endpoint, the minimum path slope is
different on the two sides, so that register cannot be correct on both.  This
contradiction gives
\[
 \dim Z=\dim\ker B_G=m-n+1.
 \tag{28}\label{eq:dag-kernel-rank}
\]
One comparison supplies at most one normal, so the execution has at least
$m-n+1$ comparisons.  Applying the projected fiber bound
\cref{cor:sp-fiber-new} to the same execution gives at least $q-d_s$
additions.  The input is a single strictly positive potential-tie weighting
for the chosen program, so the two lower bounds are simultaneous.

\begin{theorem}[Exact DAG comparison--addition frontier]
\label{thm:dag-frontier-new}
For every finite reachable directed acyclic multigraph,
\[
 \boxed{\cR(G,s)=
 \{(c,p)\in\Nzero^2:
       c\ge m-n+1,\quad p\ge q-d_s\}.}
 \tag{29}\label{eq:dag-frontier}
\]
Consequently,
\[
 C^*(G,s)=m-n+1,\qquad
 P^*(G,s)=q-d_s=\rhoF(G,s),
 \tag{30}\label{eq:dag-scalar-optima}
\]
and one program attains both scalar optima.
\end{theorem}
\begin{proof}
The topological program above attains $(m-n+1,q-d_s)$.  The generic-tie
argument gives the two coordinatewise lower bounds for every competing
program, on one input for that program.  Since the resource region is upward
closed, it is exactly the displayed rectangle.  The identity
\cref{eq:dag-rho} identifies the addition coordinate with $\rhoF$.
\end{proof}

\begin{remark}[A source-feedback corollary]
If every directed cycle contains the source, deleting self-loops and all arcs
entering $s$ leaves a reachable DAG with the same labeled distance vector.
The resource region is therefore the rectangle in \cref{eq:dag-frontier},
with $m,q,d_s$ counted in that active DAG.  Return arcs are genuine input
coordinates but are distance-irrelevant under nonnegative weights.
\end{remark}

%% file: cyclic_atom.tex
\section{The first cyclic obstruction}
\label{sec:cyclic-atom}

The DAG rectangle can fail as soon as a cycle makes two path extensions
compete for the same hub.  The smallest useful example has two spokes.

\subsection{The two-spoke shared hub}

Let $H_2$ have vertices $s,x,v,w$ and arcs
\[
 s\to x:a,
 \qquad s\to v:b,\quad v\to x:z,\quad x\to v:u,
 \qquad s\to w:c,\quad w\to x:q,\quad x\to w:p.
 \tag{31}\label{eq:h2-arcs}
\]
All seven weights are nonnegative.  A shortest walk has a simple
representative, so the distances are
\[
 d_x=\min\{a,b+z,c+q\},
 \qquad
 d_v=\min\{b,d_x+u\},
 \qquad
 d_w=\min\{c,d_x+p\}.
 \tag{32}\label{eq:h2-distances}
\]

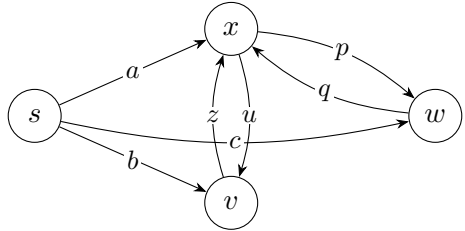
\begin{figure}[htbp]
\centering
\begin{tikzpicture}[
  vertex/.style={circle,draw,minimum size=7mm,inner sep=1pt},
  every edge/.style={draw,->,>=Stealth},
  every edge quotes/.style={font=\small,fill=white,inner sep=1pt}]
  \node[vertex] (s) at (0,0) {$s$};
  \node[vertex] (x) at (2.6,1.15) {$x$};
  \node[vertex] (v) at (2.6,-1.15) {$v$};
  \node[vertex] (w) at (5.3,0) {$w$};
  \path (s) edge["$a$"] (x)
            edge["$b$"'] (v)
        (v) edge[bend left=17,"$z$"] (x)
        (x) edge[bend left=17,"$u$"'] (v)
        (s) edge[bend right=12,"$c$"'] (w)
        (w) edge[bend left=17,"$q$"'] (x)
        (x) edge[bend left=17,"$p$"] (w);
\end{tikzpicture}
\caption{The two-spoke shared-hub atom $H_2$.  The hub distance $d_x$ is
selected from the direct route and the two inward spoke routes; it is then
reused by the two outward routes.}
\label{fig:h2}
\end{figure}

The topology in \cref{fig:h2} makes the shared hub and its two directed
feedback spokes explicit.

\begin{theorem}[Exact nonrectangularity of $H_2$]
\label{thm:h2-nonrectangular}
In the model of \cref{sec:model},
\[
 \boxed{C^*(H_2)=4,\qquad P^*(H_2)=2,\qquad C_2^*(H_2)=5.}
 \tag{33}\label{eq:h2-result}
\]
In particular, $(4,2)\notin\cR(H_2,s)$ although $(4,p)$ is achievable for
some $p$ and $(c,2)$ is achievable for some $c$.
\end{theorem}

\subsection{The two scalar upper bounds}

For $C^*(H_2)\le4$, materialize $B_z=b+z$ and $C_q=c+q$.  Two comparisons
select $X=\min\{a,B_z,C_q\}$.  Materialize $X+u$ and compare it with $b$;
materialize $X+p$ and compare it with $c$.  The output is
$(X,\min\{b,X+u\},\min\{c,X+p\})$.  This uses four comparisons and four
additions in the worst case.

For the addition coordinate, every order of $x,v,w$ after $s$ is rooted.
Exactly one of $x\to v$ and $v\to x$ is forward, and exactly one of
$x\to w$ and $w\to x$ is forward.  Hence $\rhoF(H_2,s)=2$, and
\cref{thm:arbitrary-addition-new} gives $P^*(H_2)=2$.  The following
two-addition program also makes the upper endpoint explicit.

Define $\operatorname{pick}(r,t)$ to return $r$ on $r\le t$ and $t$ on
$t<r$; it uses one comparison.  The program is:

\begin{center}
\begin{minipage}{0.94\textwidth}
\small
\begin{verbatim}
compare a with b
  if a <= b:
    compare a with c
      if a <= c:
        Au := a + u;  V := pick(Au,b)
        Ap := a + p;  W := pick(Ap,c)
        output (a,V,W)
      else:                         # c < a <= b
        Cq := c + q;  X := pick(Cq,a)
        Xu := X + u;  V := pick(Xu,b)
        output (X,V,c)
  else:                             # b < a
    compare b with c
      if b <= c:
        Bz := b + z;  T := pick(Bz,a)
        compare T with c
          if T <= c:
            Tp := T + p;  W := pick(Tp,c)
            output (T,b,W)
          else:                       # c < T
            Cq := c + q;  X := pick(Cq,T)
            output (X,b,c)
      else:                         # c < b < a
        Cq := c + q;  T := pick(Cq,a)
        compare T with b
          if T <= b:
            Tu := T + u;  V := pick(Tu,b)
            output (T,V,c)
          else:
            Bz := b + z;  X := pick(Bz,T)
            output (X,b,c)
\end{verbatim}
\end{minipage}
\end{center}

On each branch, the direct source values determine which inward candidate can
be relevant.  For example, when $b<a$ and $b\le c$, the value
$T=\min\{a,b+z\}$ satisfies $b\le T$.  If $T\le c$, then $d_x=T$ and
$d_v=b$, leaving only $\min\{c,T+p\}$; if $c<T$, then $d_w=c$ and
$b\le c\le d_x$, so $d_v=b$.  The other branches are symmetric.  Weak
inequalities make the equality choices valid for all ties and zero weights.
Every branch uses at most two additions and at most five comparisons.  Thus
\[
 C_2^*(H_2)\le5.
 \tag{34}\label{eq:h2-upper-endpoint}
\]

\subsection{A human-readable comparison lower bound}

On strictly positive inputs, the eight full-dimensional output cells have the
following distinct affine labels:
\[
\begin{array}{c|c@{\qquad}c|c}
R_1&(a,b,c)&R_2&(a,b,a+p)\\
R_3&(a,a+u,c)&R_4&(a,a+u,a+p)\\
R_5&(b+z,b,c)&R_6&(b+z,b,b+z+p)\\
R_7&(c+q,b,c)&R_8&(c+q,c+q+u,c).
\end{array}
\tag{35}\label{eq:h2-eight-cells}
\]
Each cell is nonempty; for instance, its strict inequalities can be read by
choosing the displayed winning expressions and making all unused edges large.

Suppose a correct program used at most three comparisons.  On the strict
cell interiors its comparison tree has at most eight binary leaves.  Since
the eight cells require eight different affine output labels, there would be
exactly one full-dimensional leaf per cell.  Therefore no comparison
hyperplane could cut the interior of any cell.

The following facet adjacencies form a connected graph on the eight cells:
\[
\begin{array}{lll}
R_1\!-!R_2:\ c=a+p,&R_1\!-!R_3:\ b=a+u,&R_1\!-!R_5:\ a=b+z,\\
R_1\!-!R_7:\ a=c+q,&R_2\!-!R_4:\ b=a+u,&R_2\!-!R_6:\ a=b+z,\\
R_3\!-!R_4:\ c=a+p,&R_3\!-!R_8:\ a=c+q,&R_5\!-!R_6:\ c=b+z+p,\\
R_5\!-!R_7:\ b+z=c+q,&R_7\!-!R_8:\ b=c+q+u.&
\end{array}
\tag{36}\label{eq:h2-facets}
\]
The root comparison must separate two adjacent cells, hence its homogeneous
normal is one of the seven wall normals appearing here.  In the order
\[
a=b+z,\ a=c+q,\ b=a+u,\ b=c+q+u,\ c=a+p,\ c=b+z+p,
\ b+z=c+q,
\]
these walls cut the interiors of $R_7,R_5,R_8,R_3,R_6,R_2,R_1$, respectively,
as witnessed by strictly positive points on their two sides.  Thus the root
would split a cell, contradicting the one-leaf-per-cell conclusion.  This
proves $C^*(H_2)\ge4$.

For transparency, the exact integer witnesses used for the seven wall cuts
are listed in \cref{tab:h2-wall-witnesses}.  The coordinate order is
$(a,b,c,u,z,p,q)$; the two points in each row lie in the indicated cell and
give opposite signs for the wall.

\begin{table}[htbp]
\centering
\scriptsize
\begin{tabularx}{\textwidth}{@{}lclY@{}}
\toprule
wall & cut cell & negative point & positive point\\
\midrule
$a=b+z$ & $R_7$ & $(3,1,1,1,3,1,1)$ & $(4,1,1,1,2,1,1)$\\
$a=c+q$ & $R_5$ & $(3,1,1,1,1,1,3)$ & $(4,1,1,1,1,1,2)$\\
$b=a+u$ & $R_8$ & $(4,4,1,1,1,1,1)$ & $(3,5,1,1,1,1,1)$\\
$b=c+q+u$ & $R_3$ & $(1,3,1,1,1,1,2)$ & $(1,4,1,1,1,1,1)$\\
$c=a+p$ & $R_6$ & $(4,1,4,1,1,1,1)$ & $(3,1,5,1,1,1,1)$\\
$c=b+z+p$ & $R_2$ & $(1,1,3,1,2,1,1)$ & $(1,1,4,1,1,1,1)$\\
$b+z=c+q$ & $R_1$ & $(1,1,1,1,1,1,2)$ & $(1,1,1,1,2,1,1)$\\
\bottomrule
\end{tabularx}
\caption{Integer witnesses for the wall-cut part of the $H_2$ comparison
lower bound.}
\label{tab:h2-wall-witnesses}
\end{table}

The argument is insensitive to arbitrary affine operands on a branch: on a
strict open cell, each reached register is a fixed affine form, and a
nonidentical comparison wall has empty interior in that cell.  Equality
branches cannot create an additional full-dimensional leaf.

\subsection{The two-addition obstruction}

It remains to exclude four comparisons together with only two additions.
The obstruction is a finite coefficient-state statement, not a path-register
assumption.

\begin{lemma}[Two-addition obstruction]
\label{lem:h2-two-addition-obstruction}
No deterministic program in the exact model computes labeled \DIST on $H_2$
with worst-case resources $C\le4$ and $P\le2$.
\end{lemma}
\begin{proof}[Proof outline and finite verification boundary]
Restrict to the strict positive cells in \cref{eq:h2-eight-cells} and scale
all inputs by a parameter tending to infinity.  A finite literal contributes
only a lower-order constant to every comparison; a comparison with zero
leading coefficient normal has a fixed eventual outcome and can be deleted.
Thus a hypothetical program induces a finite homogeneous coefficient-state
program with at most two binary sum insertions and four comparisons, while
allowing every sum of two available registers, repeated variables, and
non-path supports.

The state transition is exact: an addition inserts the componentwise sum of
two current coefficient vectors, and a comparison queries the difference of
two current vectors.  For each state, the eight strict output cones are
intersected with the two strict halfspaces of the query; a state is retained
only if one fixed affine output label remains correct on every surviving cone.
At an addition state all unordered pairs of currently available vectors are
enumerated.  The resulting finite case split has no surviving root state at
depth four: every comparison root either leaves two incompatible labels in
one leaf or every legal addition root spends one of the two insertions without
opening the required remaining output cell.

This is a finite proof by exhaustive closure of explicitly specified
coefficient states; it does not assume that a register is a path sum.  An
independent exact certificate checks the same closure using integer cone
certificates and is recorded as reproducibility evidence in the repository.
The finite state argument is the only part of the present lemma that remains
subject to independent human proof review.
\end{proof}

Combining \cref{eq:h2-upper-endpoint} with
\cref{lem:h2-two-addition-obstruction} gives $C_2^*(H_2)=5$, while the
previous subsections give $C^*(H_2)=4$ and $P^*(H_2)=2$.  Hence the first
cyclic example has a genuinely nonrectangular comparison--addition frontier.

%% file: shared_hub.tex
\section{The shared-hub family and endogenous sorting}
\label{sec:shared-hub}

The atom $H_2$ scales to a family in which all cyclic choices share one hub
distance.  This family isolates the resource interaction without requiring a
sorted output.

\subsection{Definition and addition optimum}

For $k\ge1$, let $H_k$ have vertices
\[
 s,x,v_1,\ldots,v_k
\]
and arcs
\[
 s\to x:a,
 \qquad s\to v_i:b_i,
 \qquad v_i\to x:z_i,
 \qquad x\to v_i:u_i
 \quad(1\le i\le k).
 \tag{37}\label{eq:hk-arcs}
\]
The exact labeled distances are
\[
 \boxed{d_x=\min\{a,b_1+z_1,\ldots,b_k+z_k\}},
 \tag{38}\label{eq:hk-hub}
\]
and
\[
 \boxed{d_i=\min\{b_i,d_x+u_i\}\qquad(1\le i\le k).}
 \tag{39}\label{eq:hk-spokes}
\]
The formulas follow by removing cycles from a shortest walk: a simple path to
$x$ uses at most one inward spoke, and a simple path to $v_i$ is either its
direct source arc or a path to $x$ followed by $x\to v_i$.

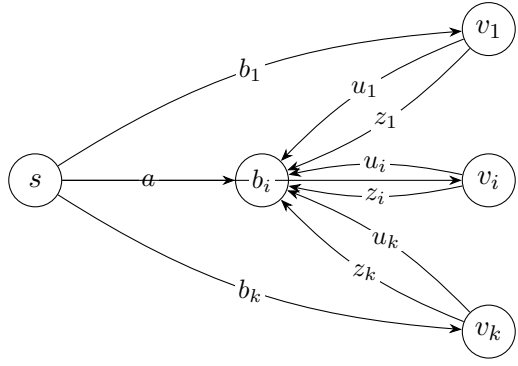
\begin{figure}[htbp]
\centering
\begin{tikzpicture}[
  vertex/.style={circle,draw,minimum size=7mm,inner sep=1pt},
  every edge/.style={draw,->,>=Stealth},
  every edge quotes/.style={font=\small,fill=white,inner sep=1pt}]
  \node[vertex] (s) at (0,0) {$s$};
  \node[vertex] (x) at (3,0) {$x$};
  \node[vertex] (v1) at (6,2.0) {$v_1$};
  \node[vertex] (vi) at (6,0) {$v_i$};
  \node[vertex] (vk) at (6,-2.0) {$v_k$};
  \path (s) edge["$a$"] (x)
        (s) edge[bend left=14,"$b_1$"'] (v1)
        (s) edge["$b_i$"'] (vi)
        (s) edge[bend right=14,"$b_k$"] (vk)
        (v1) edge[bend left=12,"$z_1$"] (x)
             edge[bend right=12,"$u_1$"'] (x)
        (vi) edge[bend left=12,"$z_i$"'] (x)
             edge[bend right=12,"$u_i$"] (x)
        (vk) edge[bend left=12,"$z_k$"'] (x)
             edge[bend right=12,"$u_k$"] (x);
\end{tikzpicture}
\caption{The shared-hub family $H_k$.  Each spoke offers one inward route to
the hub and one outward route from it.}
\label{fig:hk}
\end{figure}

The layout in \cref{fig:hk} emphasizes that all spokes share the same hub
candidate while retaining labeled spoke outputs.

Every order of $x,v_1,\ldots,v_k$ after $s$ is rooted.  For each $i$, exactly
one of $v_i\to x$ and $x\to v_i$ is forward.  Therefore
$\rhoF(H_k,s)=k$, and \cref{thm:arbitrary-addition-new} gives
\[
 \boxed{P^*(H_k)=k.}
 \tag{40}\label{eq:hk-addition-optimum}
\]

\subsection{The scalar comparison coordinate}

The scalar comparison optimum is linear.  A simple upper algorithm materializes
all $b_i+z_i$, selects the minimum with $a$, then materializes and tests each
$d_x+u_i$ against $b_i$.  It uses at most $2k$ comparisons: $k$ to select
among $k+1$ hub candidates and $k$ for the spoke outputs.

The safe scalar lower bound is the following weak form of the shared-hub
comparison lemma.

\begin{proposition}[Scalar comparison bounds]
\label{prop:hk-scalar-comparisons}
For every $k\ge1$,
\[
 \boxed{k\le C^*(H_k)\le 2k.}
 \tag{41}\label{eq:hk-scalar-bounds}
\]
In particular, $C^*(H_k)=\Theta(k)$.
\end{proposition}
\begin{proof}
The upper bound is the algorithm just described.  For the lower bound, fix
the hub distance at a raw source register and vary the $k$ spoke pairs on an
open product family so that each output is independently either its direct
register or its translated register.  At a terminal transcript the legal
materialized register for each labeled output is fixed.  The $2^k$ choice
patterns therefore require distinct strict transcripts, and a binary decision
tree has depth at least $k$.  This is only a scalar baseline; the Pareto lower
bound below uses a stronger permutation-indexed argument.
\end{proof}

\subsection{The addition-optimal endpoint}

Let $S(k)$ be the optimal worst-case number of pairwise comparisons needed to
sort $k$ distinct labeled keys.  The endpoint notation here is a special case
of \cref{eq:budget-optimum}: because $P^*(H_k)=k$, the addition-optimal
comparison value is $C_k^*(H_k)$.

\begin{theorem}[Endogenous sorting at the addition optimum]
\label{thm:hk-endogenous-sorting}
For every $k\ge1$,
\[
 \boxed{\ceil{\log_2(k!)}\le C_k^*(H_k)\le S(k)+2k.}
 \tag{42}\label{eq:hk-endpoint-sorting}
\]
Consequently,
\[
 \boxed{C_k^*(H_k)=k\log_2 k+O(k),\qquad
        C_{P^*(H_k)}^*(H_k)=\Theta(k\log k).}
 \tag{43}\label{eq:hk-endpoint-asymptotic}
\]
The graph still asks only for labeled \DIST output.
\end{theorem}
\begin{proof}
For the upper bound, sort the labeled values $b_1,\ldots,b_k$ with an
optimal comparison sorter.  A scan of the resulting order uses the
one-candidate-per-spoke rolling construction formalized in
\cref{sec:rolling}; at $r=0$ its addition count is at most $k$ and its
additional comparison count is at most $2k$.  This gives the right-hand side.

For the lower bound, the causal fan-out argument in
\cref{thm:pareto-entropy-lower}, specialized to $r=0$, shows that a strict
comparison transcript is compatible with at most one of the $k!$
scale-separated spoke orders.  A binary decision tree therefore has at least
$k!$ leaves and depth at least $\ceil{\log_2(k!)}$.  This is an
information-theoretic sorting-scale lower bound.  It does \emph{not} identify
the program's aggregate comparisons with legal pairwise key comparisons, and
so it does not claim the stronger lower bound $S(k)$.

Finally, the comparison decision-tree lower bound and a standard comparison
sort give $S(k)=k\log_2 k+O(k)$.  Combining it with the two inequalities and
\cref{eq:hk-addition-optimum} proves the asymptotic statement.
\end{proof}

The conclusion is an endogenous sorting phenomenon in the information sense.
Sorting is not requested as an output and the labeled distances need not be
emitted in any order; the arithmetic-optimal schedule nevertheless forces a
transcript that distinguishes all $k!$ spoke orders.

%% file: pareto_lower_bound.tex
\section{Lower bound for the comparison--addition tradeoff}
\label{sec:pareto-lower}

We now prove the comparison lower bound that interpolates between the linear
scalar regime and the sorting endpoint.  The proof has four logically
separate ingredients.  First, after $t$ additions at most $t$ disjoint
designated pairs can have become visible, although one addition may expose
several pairs at once.  Second, a scale-separated counterfactual keeps an adaptive
execution synchronized until the relevant pair is visible.  Third, an
addition budget $k+r$ imposes a deadline on the $r$-slack visibility schedule.
Finally, the resulting bounded-delay permutations can be counted exactly.

Throughout this section, the designated pair for spoke $i$ is
$\{b_i,z_i\}$.  For an affine coefficient vector $a$, its support is the set of
primitive variables with positive coefficient.  Since additions are sums of
nonnegative coefficient vectors, supports combine by union.

\subsection{Binary-addition pair connectivity}

\begin{lemma}[Binary-addition pair-connectivity]
\label{lem:pair-connectivity}
Let $t$ binary additions be performed on registers initially containing the
primitive variables.  For disjoint designated pairs
$\{b_i,z_i\}$, at most $t$ pairs can have appeared together in the support of
some created register.  Equivalently, if $c_i$ is the first addition time at
which a support contains both $b_i$ and $z_i$, then
\[
 \boxed{|\{i:c_i\le t\}|\le t.}
 \tag{44}\label{eq:pair-capacity}
\]
\end{lemma}
\begin{proof}
Build an undirected graph $\Gamma_t$ on the primitive variables.  Initially it
has no edges.  Maintain the invariant that the support of every materialized
register is contained in one connected component.  At an addition, the two
operand supports each lie in a current component.  If the components differ,
add one edge between arbitrary primitives in them; if they coincide, add no
edge.  The union support then lies in one component, so $\Gamma_t$ has at most
$t$ edges after $t$ additions.

Group the visible designated pairs by connected component.  A component
containing $p_j$ such pairs contains at least $2p_j$ distinct designated
terminals and therefore at least $2p_j-1\ge p_j$ edges.  Summing over
components shows that $p$ visible pairs require at least $p$ edges.  Hence
$p\le t$.  This proof permits arbitrary mixed supports, repeated operands,
and the possibility that a single late union exposes several pairs: the
preparatory unions have already paid the required edge count.
\end{proof}

\subsection{Scale-separated main and threat inputs}

The next lemma gives the scale construction explicitly.  Its small parameter
may depend on the competing program, but not on a choice made after seeing a
comparison outcome.

\begin{lemma}[Finite-program scale separation]
\label{lem:scale-separation}
Fix a deterministic finite program with finite worst-case comparison and
addition budgets, and fix a permutation
$\pi=(\pi_1,\ldots,\pi_k)$.  For all sufficiently small positive $\epsilon$
the main input
\[
 a=1,\qquad
 z_{\pi_s}=\epsilon^{4s},\qquad
 b_{\pi_s}=1-\epsilon^{4s},\qquad
 u_{\pi_s}=\epsilon^{4s+1}
 \tag{45}\label{eq:explicit-main-scales}
\]
has a strict transcript apart from comparisons that are identities on this
one-parameter family.  For each $s$, define its threat input by changing only
\[
 z_{\pi_s}\qquad\hbox{to}\qquad
 z'_{\pi_s}=\epsilon^{4s}-\epsilon^{4s+2}.
\]
Until $\{b_{\pi_s},z_{\pi_s}\}$ is visible, the main and threat executions
have identical comparison outcomes, including equality outcomes.  The claim
allows arbitrary mixed sums, repeated operands, register reuse, and finitely
many real literals.
\end{lemma}
\begin{proof}
With finite charged budgets, only finitely many formal affine register forms
can be reached: after $p$ additions their nonnegative-integer multiplicity
vectors have $\ell_1$ norm at most $2^p$, and the program has only finitely
many primitive and literal registers.  Thus it is enough to preserve a finite
set of comparison differences.

Write one such difference as $D=R-R'$, and denote the net coefficients of
$b_i,z_i,u_i$ by $\beta_{b_i},\beta_{z_i},\beta_{u_i}$.  Substitution of
\cref{eq:explicit-main-scales} turns $D$ into a polynomial in $\epsilon$:
its coefficient at exponent $4s$ is
$\beta_{z_{\pi_s}}-\beta_{b_{\pi_s}}$, its coefficient at exponent $4s+1$
is $\beta_{u_{\pi_s}}$, and all these exponents are distinct.  Literal terms
and the coefficients of $a$ contribute only to the constant coefficient.

Before the pair $\{b_i,z_i\}$ is visible, neither compared operand contains
both variables.  If
$\beta_{z_i}=\beta_{b_i}=\lambda\ne0$, then $\lambda>0$ would force $R$ to
contain both variables, while $\lambda<0$ would force $R'$ to contain both.
Consequently, whenever $\beta_{z_i}\ne0$ we have
$\beta_{z_i}-\beta_{b_i}\ne0$.  On the threat input,
\[
 D_{\rm threat}(\epsilon)-D_{\rm main}(\epsilon)
   =-\beta_{z_i}\epsilon^{4s+2}.
\]
If $\beta_{z_i}\ne0$, the main polynomial already has a nonzero term at
exponent $4s$, so the order-$4s+2$ change cannot alter its sign for
sufficiently small $\epsilon$.  If $\beta_{z_i}=0$, the comparison value does
not change.  If the main polynomial is identically zero, the same argument
forces $\beta_{z_i}=0$, so equality is preserved exactly.

Every nonzero polynomial has a constant sign after deleting a sufficiently
small interval at the origin and its finitely many positive roots.  Taking a
common $\epsilon$ for the finite family of possible differences proves the
claim.  We also include in this finite family the differences between every
possible held form and the finitely many threat-output forms used below; this
only deletes finitely many additional roots.  The displayed main and threat
weights are nonnegative for $0<\epsilon<1$.
\end{proof}

\begin{lemma}[Synchronization before pair visibility]
\label{lem:threat-synchronization}
Let the main input and a threat input be chosen as in
\cref{lem:scale-separation}.  Until the threatened pair
$\{b_i,z_i\}$ first occurs together in the support of a materialized
register, the two executions have the same control state, execute identical
instructions, have identical coefficient-vector states, and return identical
comparison outcomes.
\end{lemma}
\begin{proof}
Initially the executions agree on all already committed leading scales.  If
they have agreed through one instruction, a free operation or an addition
applies to the same registers and therefore creates the same coefficient
vector.  Before the pair is visible, the polynomial calculation in
\cref{lem:scale-separation} preserves every strict outcome and every equality
outcome, so the next instruction is the same.
Induction over the finite transcript proves the claim.  In particular, the
argument is about actual materialized-register comparisons and does not turn a
counterfactual into a model operation.
\end{proof}

\subsection{Threat deadlines}

Fix a permutation $\pi=(\pi_1,\ldots,\pi_k)$ of the spokes and use the main
input in \cref{eq:explicit-main-scales}.  Its thresholds
$b_{\pi_s}-u_{\pi_s}=1-\epsilon^{4s}-\epsilon^{4s+1}$ are strictly increasing
with $s$.  Let $c_i$ be the first visibility time from
\cref{lem:pair-connectivity}.

\begin{lemma}[Threat deadline]
\label{lem:threat-deadline}
If a correct program for $H_k$ has worst-case addition budget at most $k+r$,
then for every main permutation $\pi$ and every position $1\le s\le k$,
\[
 \boxed{c_{\pi_s}\le s+r.}
 \tag{46}\label{eq:threat-deadline}
\]
\end{lemma}
\begin{proof}
Assume $c_{\pi_s}>s+r$ and inspect the execution after
$t=s+r$ additions.  Apply the synchronized threat construction to the hidden
pair $Q_{\pi_s}=\{b_{\pi_s},z_{\pi_s}\}$.  On the threat input the unique
shortest path to the hub $x$ uses this inward spoke and has value
$1-\epsilon^{4s+2}$.  For $j\le s$, the direct source arc is uniquely
shortest to $v_{\pi_j}$; for $j>s$, the unique shortest path is the threatened
hub path followed by $x\to v_{\pi_j}$.  Indeed, these alternatives are
decided by comparing $\epsilon^{4j}+\epsilon^{4j+1}$ with
$\epsilon^{4s+2}$.  Thus the
\[
 M_s=1+(k-s)=k-s+1
\]
required labeled output values are
\[
 1-\epsilon^{4s+2},\qquad
 1-\epsilon^{4s+2}+\epsilon^{4j+1}\quad(j>s).
\]
They are pairwise distinct, and every displayed form contains both variables
of $Q_{\pi_s}$.

No held form whose support omits one member of $Q_{\pi_s}$ can equal one of
these values.  To see this, regard the evaluation as a polynomial in
$\epsilon$.  The coefficient at exponent $4s+2$ forces one copy of the
threatened $z_{\pi_s}$, and cancellation of its exponent-$4s$ term then
forces one copy of $b_{\pi_s}$.  Distinct exponents exclude cancellation by
another primitive, while literals affect only the constant term.  Our common
choice of $\epsilon$ avoids every accidental root of the finitely many
nonidentity differences.

By synchronization, the two executions have the same instructions and
coefficient supports through the prefix.  Because the pair is still hidden,
none of these $M_s$ target values has been materialized.  Free copy and
selection create no new numerical value, while each future binary addition
creates only one register.  The remaining budget is
\[
 (k+r)-(s+r)=k-s=M_s-1,
\]
so the threat execution cannot materialize all required outputs.  This
contradicts correctness.  If the main branch terminates before $t$, the same
threat follows its complete transcript and already yields the contradiction.
Therefore $c_{\pi_s}\le s+r$.
\end{proof}

The deadline lemma is the point at which the addition budget becomes an
organizational constraint.  It does not say that the program outputs an
order; it says that a successful addition-tight execution must expose spoke
interfaces with bounded delay relative to the counterfactual order.

\subsection{The exact banded-permutation count}

\begin{lemma}[Banded permutation count]
\label{lem:banded-count}
Let $B(k,r)$ be the number of permutations $\sigma$ of $[k]$ satisfying
\[
 \sigma_s\le s+r\qquad(1\le s\le k).
 \tag{47}\label{eq:band-condition}
\]
Then
\[
 \boxed{B(k,r)=r!(r+1)^{k-r}.}
 \tag{48}\label{eq:banded-count}
\]
\end{lemma}
\begin{proof}
For position $s\le k-r$, the allowed values lie in
$\{1,\ldots,s+r\}$.  The $s-1$ values already placed all lie in
$\{1,\ldots,s+r-1\}$, so exactly $r+1$ allowed values remain.  After these
$k-r$ positions have been filled, the remaining $r$ values can be placed in
arbitrary order.  Multiplying gives $(r+1)^{k-r}r!$.  The same calculation
also covers $r=k$, when the first product is empty.
\end{proof}

\subsection{Entropy lower bound}

Define
\[
 F_k(r)=C_{k+r}^*(H_k),qquad
 \Lambda_{k,r}=\log_2\frac{k!}{r!(r+1)^{k-r}}.
 \tag{49}\label{eq:lower-lambda}
\]

\begin{theorem}[Banded-permutation entropy lower bound]
\label{thm:pareto-entropy-lower}
Uniformly for $0\le r\le k$,
\[
 \boxed{F_k(r)\ge\ceil{\Lambda_{k,r}}.}
 \tag{50}\label{eq:entropy-lower}
\]
Together with the scalar baseline,
\[
 \boxed{F_k(r)\ge
 \max\left\{k,\ceil{\Lambda_{k,r}}\right\}.}
 \tag{51}\label{eq:combined-lower}
\]
\end{theorem}
\begin{proof}
Fix a complete strict comparison transcript.  Its deterministic addition
trace fixes all exposure times $c_i$.  Rank the pairs by nondecreasing
exposure time, breaking ties by label, and write $\rho(i)$ for the rank of
pair $i$.  The cumulative capacity lemma gives $\rho(i)\le c_i$.  Hence every
main permutation $\pi$ compatible with this transcript satisfies
\[
 \rho(\pi_s)\le c_{\pi_s}\le s+r.
\]
After relabeling by $\rho$, \cref{lem:banded-count} shows that this transcript
is compatible with at most $B(k,r)$ permutations.  Choose the finite main
family generically after the program is fixed, so every reached nonidentity
comparison is strict; forced identities have only one outcome.  A depth-$C$
program has at most $2^C$ relevant transcripts.  Consequently
\[
 k!\le 2^C B(k,r).
\]
Substitute \cref{eq:banded-count} and take logarithms.  The scalar baseline is
the independent bound in \cref{prop:hk-scalar-comparisons}; taking the maximum
gives \cref{eq:combined-lower}.
\end{proof}

The proof is uniform in $r$: the program, its finite transcript, and its
literal constants are fixed first, and the scale construction is then chosen
for that transcript.  The finite support search and scale sanity checks used
for reproduction are evidence for the lemmas, not substitutes for their
quantified proofs.

%% file: bbwo.tex
\section{Bounded-bucket weak ordering}
\label{sec:bbwo}

The upper bound needs only a weak order of the spoke values.  We state the
needed comparison lemma separately so that the construction is not confused
with the paper's SSSP contribution.

\begin{definition}[Bounded-bucket weak ordering]
\label{def:bbwo}
Given $k$ keys and an integer $h\ge1$, a bounded-bucket weak ordering is a
sequence of blocks
\[
 B_1< B_2<\cdots<B_m,
 \qquad |B_j|\le h,
 \tag{52}\label{eq:bbwo}
\]
where every key in an earlier block is smaller than every key in a later
block.  The internal order of a block is not produced.  Let $Q(k,h)$ be the
minimum worst-case number of pairwise comparisons needed to produce such
blocks.  On inputs with ties, equality can be assigned to a deterministic
adjacent block; the entropy calculation below is for distinct keys.
\end{definition}

For $k=qh+s$ with $0\le s<h$, define
\[
 \Gamma_{k,h}=\log_2\frac{k!}{(h!)^q s!}.
 \tag{53}\label{eq:gamma}
\]

\begin{lemma}[Bounded-bucket weak-order complexity]
\label{lem:bbwo}
For $1\le h\le k$,
\[
 \boxed{\Gamma_{k,h}\le Q(k,h)
 \le \Gamma_{k,h}+o(\Gamma_{k,h})+O(k).}
 \tag{54}\label{eq:bbwo-complexity}
\]
If $h\ge k$, one block suffices and $Q(k,h)=0$.  In particular, for
$h=k+1$, both the natural entropy expression and the comparison cost are
zero.
\end{lemma}
\begin{proof}
For the lower bound, consider the decision tree on distinct keys.  A leaf
whose output blocks have sizes $t_1,t_2,\ldots$ is compatible with at most
$\prod_j t_j!$ total orders, because only the order inside each block is
unknown.  Subject to $t_j\le h$ and $\sum_jt_j=k$, this product is maximized
by filling blocks of size $h$ and one residual block of size $s$; factorial
supermultiplicativity gives
\[
 \prod_jt_j!\le(h!)^q s!.
\]
At least $k!/((h!)^q s!)$ leaves are therefore needed, giving the lower
bound $Q(k,h)\ge\Gamma_{k,h}$.

For the upper bound, request the ranks $h,2h,\ldots,qh$ from the deterministic
multiple-selection algorithm of Kaligosi, Mehlhorn, Munro, and
Sanders~\cite{kaligosi2005multiple}.  Their worst-case bound is the
information lower bound for the requested ranks, plus a lower-order term in
that bound and $O(k)$.  The selected splitters and the induced intervals are
exactly the ordered blocks in \cref{eq:bbwo}; their information term is
$\Gamma_{k,h}$.  This gives the stated upper bound.  The case $h\ge k$ is
immediate.
\end{proof}

\begin{remark}[Scope of the lemma]
Producing blocks is a standard weak-order task.  Multiple selection gives the
deterministic upper bound above; partial-order production provides a related
$\mathrm{ITLB}+o(\mathrm{ITLB})+O(k)$ framework
\cite{kaligosi2005multiple,cardinal2010partial}.  We do not replace either
lower-order term by $O(k)$: such a deterministic strengthening is not supplied
by these references.  The contribution here is only the interface from this
standard primitive to the rolling SSSP computation.
\end{remark}

%% file: rolling_block.tex
\section{Rolling-block transfer to SSSP}
\label{sec:rolling}

Let $0\le r\le k$ and set $h=r+1$.  Assume that the spoke values have been
placed in ordered blocks
\[
 B_1<\cdots<B_m,\qquad |B_j|\le h,
 \tag{57}\label{eq:rolling-blocks}
\]
by the bounded-bucket procedure.  We show how to turn this weak order into a
legal exact \DIST computation with at most $k+r$ additions.

\subsection{The algorithm}

Maintain a current hub candidate $H$, initially the source-to-hub register
$a$.  Scan the blocks from left to right.  Within a block, compare each
$b_i$ with the current $H$.  If $b_i<H$, materialize the inward candidate
$I_i=b_i+z_i$ and replace $H$ by the smaller of $H$ and $I_i$.  If a scan
encounters a spoke with $b_i\ge H$, stop the inward phase after finishing
that block; call this the crossing block.  Spokes in later blocks are not
given inward candidates.

Once the hub candidate is final, output a spoke directly when $b_i\le H$.
Otherwise materialize the outward candidate $O_i=H+u_i$ and compare it with
$b_i$.  All comparisons use a fixed equality rule.  The following schematic
figure, \cref{fig:rolling-block}, records the only block in which both kinds
of candidates can occur.

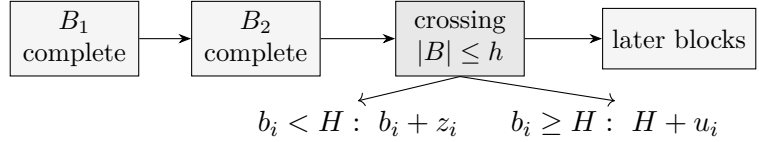
\begin{figure}[htbp]
\centering
\begin{tikzpicture}[
  box/.style={draw,thin,minimum height=8mm,minimum width=17mm,
              align=center,font=\small},
  every edge/.style={draw,thin,->,>=Stealth}]
  \node[box,fill=gray!8] (b1) at (0,0) {$B_1$\\complete};
  \node[box,fill=gray!8] (b2) at (2.4,0) {$B_2$\\complete};
  \node[box,fill=gray!18] (cross) at (5.1,0) {crossing\\$|B|\le h$};
  \node[box,fill=gray!8] (later) at (8.0,0) {later blocks};
  \node (low) at (3.75,-1.12) {$b_i<H:\ b_i+z_i$};
  \node (high) at (7.15,-1.12) {$b_i\ge H:\ H+u_i$};
  \path (b1) edge (b2) (b2) edge (cross) (cross) edge (later);
  \draw[->,thin] (cross.south) -- (low.north);
  \draw[->,thin] (cross.south) -- (high.north);
\end{tikzpicture}
\caption{Rolling-block processing.  Complete blocks are permanently below
the current hub candidate.  Only the crossing block can contain a spoke that
receives both an inward and an outward candidate.}
\label{fig:rolling-block}
\end{figure}

\subsection{Correctness and addition accounting}

\begin{lemma}[Rolling-block transfer]
\label{lem:rolling-transfer}
Given ordered blocks satisfying \cref{eq:rolling-blocks}, the rolling algorithm
computes all labeled distances of $H_k$ correctly and uses at most
\[
 \boxed{P\le k+r}
 \tag{58}\label{eq:rolling-additions}
\]
additions.  Its additional comparison cost is $O(k)$.
\end{lemma}
\begin{proof}
At every inward step, $H$ is the minimum of $a$ and all inward candidates
materialized so far.  If a spoke has $b_i\ge H$, then
$b_i+z_i\ge H$ because $z_i\ge0$.  Since later blocks contain only larger
$b$ values, no later block can lower $H$.  Thus the hub becomes final at the
first crossing block, if one exists; if no crossing occurs, it is final after
the last block.

For a complete block, every $b_i$ in the block remains at most the final
$H$: the block was scanned while all its values were below the current
candidate, and every later update is a minimum of candidates no smaller than
the relevant direct value.  Hence
$d_i=\min\{b_i,H+u_i\}=b_i$ for those spokes, and their direct outputs are
permanently safe.  In the crossing block, spokes below the current $H$ may
receive inward candidates; after $H$ is finalized, only those whose direct
value is above $H$ need an outward candidate.  At least one spoke in this
block triggered the crossing and therefore received no inward candidate.  A
spoke in a later block receives no inward candidate at all.

Charge one candidate addition to every spoke: an inward candidate for a
complete or below-threshold crossing spoke, and an outward candidate for a
crossing or later spoke that is above the final hub.  The only possible second
addition is an outward candidate for a below-threshold spoke in the crossing
block.  There are at most $|B|-1\le h-1=r$ such spokes, proving
$P\le k+r$.  The scan uses only a constant number of comparisons per spoke,
so its comparison cost is $O(k)$.  The weak inequalities also cover ties and
zero weights.
\end{proof}

\begin{proposition}[Weak order to Pareto upper bound]
\label{prop:rolling-upper}
For $0\le r\le k$ and $h=r+1$,
\[
 C_{k+r}^*(H_k)\le Q(k,h)+O(k).
 \tag{59}\label{eq:rolling-upper}
\]
\end{proposition}
\begin{proof}
Produce the blocks with a BBWO algorithm and apply
\cref{lem:rolling-transfer}.  The budget is $k+r$ and all numerical values
that are output have been materialized by the counted additions.
\end{proof}

%% file: pareto_law.tex
\section{The entropy-tight comparison--addition law}
\label{sec:pareto-law}

The lower bound in \cref{thm:pareto-entropy-lower} and the rolling upper
bound have the same leading entropy term whenever that term dominates the
linear scalar floor.  The comparison between the two entropy expressions is
elementary but useful because it translates the
weak-order construction into the exact budget parameter from
\cref{eq:budget-optimum}.

\subsection{Comparing the two entropy terms}

Let $h=r+1\le k$, and write $k=qh+s$ with $0\le s<h$.  Then
\[
 \Gamma_{k,h}-\Lambda_{k,r}
 =\log_2\frac{h^{(q-1)(h-1)+s}}
                    {((h-1)!)^{q-1}s!}.
 \tag{60}\label{eq:gamma-lambda-difference}
\]

\begin{lemma}[Entropy comparison]
\label{lem:entropy-comparison}
For $h=r+1\le k$,
\[
 \boxed{0\le \Gamma_{k,h}-\Lambda_{k,r}
             \le (\log_2 e)k.}
 \tag{61}\label{eq:entropy-comparison}
\]
For $h=k+1$ (equivalently $r=k$), both entropy quantities are zero.
\end{lemma}
\begin{proof}
The ratio in \cref{eq:gamma-lambda-difference} is at least one because
$(h-1)!\le h^{h-1}$ and $s!\le h^s$.  For the upper bound, the elementary
factorial estimate
\[
 \frac{h^{h-1}}{(h-1)!}\le e^h
 \tag{62}\label{eq:factorial-estimate}
\]
follows, for example, from the integral form of the logarithm bound for
factorials.  Since $0\le s<h$, the ratio $h^s/s!$ is maximized at
$s=h-1$, so it is also at most $e^h$.  Thus the ratio is at most
$e^{qh}\le e^k$.  Taking base-two logarithms proves the claim.  If
$h=k+1$, then $q=0$, $s=k$, and both denominators equal $k!$.
\end{proof}

\subsection{The complete budget law}

The envelope in \cref{fig:pareto-law} summarizes the three asymptotic regimes
of the tradeoff.

\begin{figure}[htbp]
\centering
\begin{tikzpicture}[x=1.0cm,y=0.72cm,font=\small]
  \draw[-{Stealth[length=2mm]},thin] (0,0) -- (10.6,0)
        node[below=4pt] {$r$ in the addition budget $k+r$};
  \draw[-{Stealth[length=2mm]},thin] (0,0) -- (0,5.7)
        node[left=6pt,rotate=90] {comparison complexity};
  \draw[black,thick]
        plot[smooth] coordinates {(0.8,5.0) (2.1,4.48) (3.6,3.82)
                                  (5.2,3.12) (6.9,2.45) (8.5,1.82)
                                  (9.8,1.22)};
  \fill[black] (0.8,5.0) circle (1.8pt);
  \fill[black] (5.2,3.12) circle (1.8pt);
  \fill[black] (9.8,1.22) circle (1.8pt);
  \node[align=left,anchor=west] at (1.0,5.15)
        {$r=0$\\$\Theta(k\log k)$};
  \node[align=center,fill=white,inner sep=1.5pt] at (5.4,3.75)
        {$\Lambda_{k,r}=\log_2\!\dfrac{k!}{r!(r+1)^{k-r}}$};
  \node[align=right,anchor=east] at (9.6,0.7)
        {$r=k$\\$\Theta(k)$};
  \draw[densely dashed,gray] (0.8,5.0) -- (0.8,0);
  \draw[densely dashed,gray] (9.8,1.22) -- (9.8,0);
\end{tikzpicture}
\caption{Comparison cost on $H_k$ as the addition slack grows.  The entropy
term controls the curve when it dominates $k$; the scalar lower bound leaves
a linear floor near $r=k$.  The drawing is schematic.}
\label{fig:pareto-law}
\end{figure}
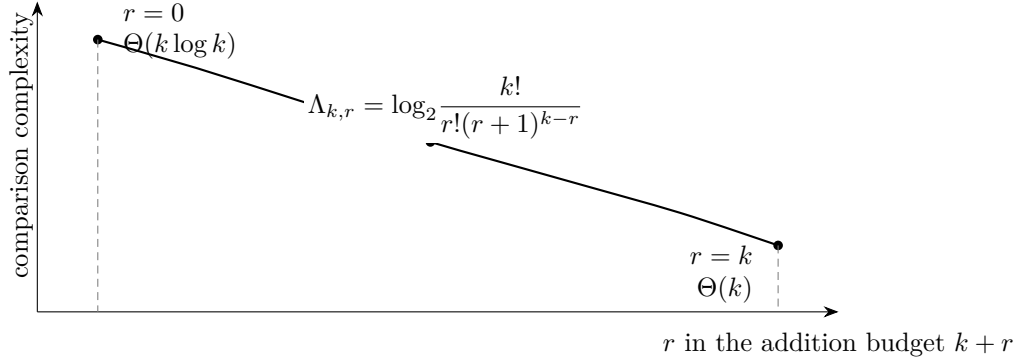

\begin{theorem}[Deterministic entropy-tight Pareto law]
\label{thm:near-exact-pareto}
For every $k\ge1$ and every $0\le r\le k$, let
\[
 F_k(r)=C_{k+r}^*(H_k),
 \qquad
 \Lambda_{k,r}=\log_2\frac{k!}{r!(r+1)^{k-r}}.
 \tag{63}\label{eq:pareto-law-def}
\]
Uniformly in $r$,
\[
 \boxed{
 \max\left\{k,\ceil{\Lambda_{k,r}}\right\}
 \le F_k(r)
 \le \Gamma_{k,r+1}+o(\Gamma_{k,r+1})+O(k),}
 \tag{64}\label{eq:pareto-law-explicit}
\]
where $\Gamma_{k,k+1}=0$.  Consequently,
\[
 \boxed{F_k(r)=\Theta(k+\Lambda_{k,r})}
 \tag{65}\label{eq:pareto-law}
\]
uniformly over $0\le r\le k$; moreover
$F_k(r)=(1+o(1))\Lambda_{k,r}$ along every asymptotic regime in which
$\Lambda_{k,r}/k\to\infty$.
\end{theorem}
\begin{proof}
The lower inequality is \cref{thm:pareto-entropy-lower}, including the
scalar baseline.  For the upper bound, set $h=r+1$.  Produce ordered blocks
with \cref{lem:bbwo} and transfer them to $H_k$ using
\cref{prop:rolling-upper}.  This gives
\[
 F_k(r)\le \Gamma_{k,h}+o(\Gamma_{k,h})+O(k).
 \]
When $h\le k$, apply \cref{lem:entropy-comparison}; the difference from
$\Lambda_{k,r}$ is at most $(\log_2e)k$.  When $h=k+1$, one block is emitted,
$\Lambda_{k,k}=0$, and the rolling computation costs $O(k)$ comparisons.
The entropy comparison gives $\Gamma_{k,h}=\Lambda_{k,r}+O(k)$.  Together
with the scalar lower bound this proves the uniform $\Theta$ statement.  If
$\Lambda_{k,r}/k\to\infty$, both the $O(k)$ term and the multiple-selection
lower-order term are $o(\Lambda_{k,r})$, proving the final assertion.
\end{proof}

At $r=0$, $\Lambda_{k,0}=\log_2(k!)$, so the law recovers the sorting-scale
endpoint in \cref{thm:hk-endogenous-sorting}.  At $r=k$, the entropy term
vanishes and the scalar linear baseline controls the comparison value.  The
intermediate levels are not asserted to have an additive-$O(k)$ deterministic
correction: the standard deterministic multiple-selection bound contains an
additional lower-order entropy term.  No statement here is intended for
arbitrary directed topologies.

%% file: universal_optimality.tex
\section{Exact numerical universal optimality}
\label{sec:universal-numerical}

We now ask whether topology specialization confers an intrinsic numerical
advantage.  The answer is no: one fixed interpreter can match \(\OPT(G,s)\)
exactly on every topology.  The interpreter is not efficient; the result
separates existence of an optimal numerical policy from the cost of navigating
that policy.

\subsection{Homogeneous programs and transcript cones}

Positive homogeneity of shortest-path distances lets us remove arbitrary
literal constants without weakening a competitor.

\begin{lemma}[Recession normal form]
\label{lem:recession-normal-form}
For every correct finite \CA program \(A\), there is a correct program
\(A^\infty\) using only the literal zero such that
\[
 C_{A^\infty}(G,s,w)+P_{A^\infty}(G,s,w)
 \le T_A(G,s)
 \quad\text{for every }w\ge0,
 \tag{68}\label{eq:recession-bound}
\]
and every held register of \(A^\infty\) has form
\(\langle a,w\rangle\) for \(a\in\Nzero^m\).
\end{lemma}
\begin{proof}[Proof idea]
On each unfolded branch, replace the affine register
\(\langle a,w\rangle+\alpha\) by its leading homogeneous form.  Route a
comparison by the sign of its leading difference; when that difference is
zero, route by the fixed intercept sign.  For any fixed \(w\), the resulting
trace is the stable trace of \(A(tw)\) for all sufficiently large \(t\).
Homogeneity of the output then identifies each leading output form with the
required distance.  The complete finiteness and equality-face argument is in
\cref{app:recession}.
\end{proof}

Let \(M\) be the set of homogeneous coefficient forms represented by
physical held registers.  Initially
\[
 M_0=\{0,e_1,\ldots,e_m\}.
 \tag{69}\label{eq:initial-forms}
\]
A physical comparison transcript is
\[
 \Gamma=((a_i,b_i,\sigma_i))_{i=1}^t,
 \qquad \sigma_i\in\{<,=,>\},
\]
and represents the relatively open rational polyhedral cone
\[
 \mathcal P(\Gamma)=
 \{w\in\R_{\ge0}^m:
   \operatorname{sign}\langle a_i-b_i,w\rangle=\sigma_i
   \text{ for every }i\}.
 \tag{70}\label{eq:transcript-cone}
\]
Only feasible transcripts are retained.

\begin{definition}[Ready state]
\label{def:ready}
The pair \((M,\Gamma)\) is \(\Ready\) if, for every output label
\(v\ne s\), there is one fixed form \(z_v\in M\) such that
\[
 \langle z_v,w\rangle=d_w(s,v)
 \qquad\text{for every }w\in\mathcal P(\Gamma).
 \tag{71}\label{eq:ready}
\]
The witness may depend on \(v\), but it may not vary over an unobserved
subregion of the transcript cone.
\end{definition}

\subsection{The exact total-budget game}

Define \(\WIN_G(M,\Gamma,r)\) recursively.  A ready state is winning.  A
nonready state with \(r=0\) is losing.  Otherwise the policy may choose one of
two moves:

\begin{itemize}[leftmargin=*,itemsep=2pt]
  \item add \(x,y\in M\), including \(x=y\), insert \(x+y\), and decrease
        \(r\) by one; or
  \item compare distinct represented forms \(x,y\), decrease \(r\) by one,
        and win for every feasible outcome in \(\{<,=,>\}\).
\end{itemize}

An already represented sum and a comparison with only one feasible outcome
may be omitted.  The game state records only numerical capabilities and the
physical outcomes actually observed; it does not store a chamber arrangement
or an unfolded strategy tree.

\begin{theorem}[Exact transcript-cone game]
\label{thm:exact-transcript-game}
For every finite reachable rooted directed multigraph and every integer
\(K\ge0\),
\[
 \boxed{
 \WIN_G(M_0,\varnothing,K)
 \quad\Longleftrightarrow\quad
 \OPT(G,s)\le K.}
 \tag{72}\label{eq:exact-game}
\]
\end{theorem}
\begin{proof}
Apply \cref{lem:recession-normal-form} to an arbitrary specialized program.
Retain every materialized form, delete duplicate additions, and bypass
comparisons whose outcomes are implied by the current physical transcript.
An addition does not restrict the compatible weights; a comparison intersects
them with exactly its observed sign condition.  Thus the inputs reaching a
node are precisely \(\mathcal P(\Gamma)\), every charged instruction is a
legal move, and correctness at a leaf supplies the fixed witnesses required
by \cref{def:ready}.  Backward induction yields a winning strategy with no
larger total budget.

Conversely, select a winning move at each winning state and maintain one
physical register for every form in \(M\).  Execute additions and comparisons
literally and follow the observed ternary outcome.  At a ready state output
the fixed witness registers.  Their validity over the whole transcript cone
includes equality faces, the origin, parallel primitives, ties, and
zero-weight cycles.  Each play uses at most \(K\) charged operations, giving a
correct specialized program.  The two simulations prove the equivalence.
\end{proof}

\begin{corollary}[No numerical uniformity tax]
\label{cor:no-numerical-tax}
There is one fixed deterministic uniform interpreter \(U_{\rm pspace}\) such
that for every \((G,s)\),
\[
 \boxed{T_{U_{\rm pspace}}(G,s)=\OPT(G,s).}
 \tag{73}\label{eq:no-numerical-tax}
\]
\end{corollary}

The interpreter first finds the least winning budget from the topology.  It
then recomputes a winning move from the current state, executes exactly that
physical addition or comparison on the input, and appends the observed
outcome.  Equality in \cref{eq:no-numerical-tax}, rather than only an upper
bound, follows from the definition of \(\OPT\).

\subsection{Succinct state and polynomial space}

A persistent implementation represents materialization by an addition DAG:
primitive and zero registers are leaves, and each executed addition stores two
parent identifiers.  Each comparison stores two identifiers and its ternary
outcome.  If two registers have the same coefficient vector, one
representative suffices because their future numerical capabilities coincide.

Bellman--Ford gives the legal total-cost cap
\[
 K_0=2m(n-1).
 \tag{74}\label{eq:bellman-cap}
\]
Along a length-\(K\) play there are at most \(m+K+1\) forms.  After \(j\)
additions, every coefficient is at most \(2^j\), hence has \(O(j)\) bits.
The persistent DAG and transcript use \(O(m+n+K)\) constant-size records;
expanded coefficient vectors are temporary polynomial-size objects.

Transcript feasibility is rational linear programming.  The terminal
predicate has a short complementary certificate: to show non-readiness,
choose one output label and, for every held form, give a possibly different
transcript-compatible rational weighting and a simple shortest path on which
that form is wrong.

\begin{theorem}[Polynomial-space exact policy generation]
\label{thm:pspace-policy}
The decision problem
\[
 (G,s,K)\longmapsto[\OPT(G,s)\le K]
\]
is in deterministic PSPACE.  The exact value of \(\OPT\) and a winning next
action at every reached transcript-cone state are computable in polynomial
space.  A direct depth-first implementation uses
\(2^{\operatorname{poly}(n+m)}\) deterministic time.
\end{theorem}
\begin{proof}
For \(K\ge K_0\), the Bellman--Ford program answers yes.  Otherwise the game
has polynomial depth and polynomially many candidate actions per succinct
state.  Enumerate actions one at a time, discard infeasible comparison
outcomes by rational linear programming, test \(\Ready\) in coNP, and evaluate
the alternating recurrence depth first.  The recursion depth, state encoding,
and temporary certificates are polynomial.  Iterating
\(K=0,\ldots,K_0\) gives the optimum, and re-enumerating the actions gives a
winning next move.  The detailed terminal certificate is proved in
\cref{app:ready-pspace}.
\end{proof}

The polynomial state does not imply a small unfolded policy.  For \(r\)
independent output vertices, each with two parallel source primitives,
\(\OPT=r\), but the \(2^r\) strict winner vectors require \(2^r\) distinct
fixed-output leaves.  The compact loop and transcript state avoid storing that
tree; they do not make next-action search fast.

The hierarchy is therefore strict: the transcript-cone interpreter has ratio
one in charged operations; its optimal-action navigation uses polynomial
space but may take exponential time; the active-core lifting sacrifices a
sublogarithmic factor to obtain an efficient uniform algorithm.

%% file: efficient_universality.tex
\section{Efficient universal upper bound}
\label{sec:efficient-universality}

The ratio-one interpreter may spend exponential ordinary time between two
charged operations.  We next give the strongest current guarantee under
linear topology preprocessing and ordinary RAM accounting.  The bridge is a
canonical active core whose size is controlled by the full same-program
benchmark.

\subsection{Canonical active core}

An arc \(e=(u,v)\) is \emph{dominator-inactive} if \(v\) dominates \(u\),
that is, every \(s\)-to-\(u\) path passes through \(v\).  Delete all such
arcs.  A simple rooted path using \((u,v)\) has already visited \(v\), so that
arc closes a nonnegative subwalk; deleting the subwalk cannot increase a
distance.  Conversely, setting all deleted primitives to zero simulates any
program for the original graph on the active graph.  Inactive deletion is
therefore exactly benchmark-neutral.

Among the remaining source arcs, let \(p_v\) be the multiplicity of the class
\(s\to v\), and define
\[
 m_s=\sum_v p_v,\qquad
 d_s=|\{v:p_v>0\}|,\qquad
 \sigma=m_s-d_s.
 \tag{75}\label{eq:sigma}
\]
Reducing these classes costs exactly \(\sigma\) comparisons and yields held
minimum registers \(a_v\).  Let \(E_\circ\) be the active primitives with
nonsource tail, put \(k=|E_\circ|\), and let \(R\) be their endpoints.

The canonical core \(H=H(G,s)\) has source \(q\), vertex set
\(\{q\}\cup R\), every primitive of \(E_\circ\), and one virtual source arc
\(q\to v\) of weight \(a_v\) whenever such a register exists.  Register
aliasing is free and legal.  Every vertex outside \(R\) is source-only and its
distance is already \(a_v\).  Moreover,
\[
 |V(H)|\le2k+1,\qquad |E(H)|\le3k.
 \tag{76}\label{eq:core-size}
\]

\begin{theorem}[Active-core benchmark factorization]
\label{thm:active-core-factorization}
For every finite reachable rooted directed multigraph,
\[
 \boxed{
 \OPT(H,q)\le\OPT(G,s)\le\sigma+\OPT(H,q)\le2\OPT(G,s).}
 \tag{77}\label{eq:active-core-factorization}
\]
If \(\sigma=0\), then \(\OPT(G,s)=\OPT(H,q)\).
\end{theorem}
\begin{proof}
For the first inequality, simulate any program for \(G\) on an arbitrary core
input: alias every source class entering \(R\) to its single virtual register,
set inactive and discarded source-only primitives to zero, run the original
program, and keep only coordinates in \(R\).  This adds no charged operation
and preserves all mixed sums, equality branches, and cross-output sharing.

For the second inequality, reduce source classes in \(\sigma\) comparisons,
run a program for \(H\) on the held minima and internal primitives, and return
the already held source-only outputs.  The virtual-source path identity makes
the outputs exact.

Finally, the same-input active-edge charge in
\cref{lem:active-edge-charge} gives \(\sigma+k\le2\OPT(G,s)\), in particular
\(\sigma\le\OPT(G,s)\).  Combining this with
\(\OPT(H,q)\le\OPT(G,s)\) gives the last inequality.  When
\(\sigma=0\), the first two inequalities coincide.
\end{proof}

The factorization is not a sum of separate comparison and addition lower
bounds.  Both sides use \(\OPT\), and the active-edge charge is witnessed on
one execution of each competing program.

\subsection{Lifting a general directed-SSSP algorithm}

Duan, Mao, Shu, and Yin give a deterministic exact directed-SSSP algorithm on
\(M\) arcs and \(N\) vertices with nonnegative real weights and running time
\cite{duan2026faster}
\[
 F(M,N)=O\!\left(
 M\sqrt{\log N}+\sqrt{MN\log N\log\log N}
 \right).
 \tag{78}\label{eq:duan-bound}
\]
Their algorithm is used here as a black box; the topology-sensitive statement
comes from the active-core reduction and benchmark charge.

\begin{theorem}[Efficient topology-wise bound]
\label{thm:efficient-universal}
There is one explicit deterministic uniform exact \DIST algorithm with
\(O(m+n)\) topology preprocessing and
\[
 \boxed{
 T_U(G,s)=O\!\left(
 \OPT(G,s)
 \sqrt{
   \log_+(2+\OPT(G,s))
   \log\log_+(4+\OPT(G,s))}
 \right).}
 \tag{79}\label{eq:efficient-universal}
\]
Its ordinary query work is
\(O(m+n+C_U(G,s,w)+P_U(G,s,w))\).
\end{theorem}
\begin{proof}
Reduce the active source classes and run the algorithm underlying
\cref{eq:duan-bound} on \(H\).  By \cref{eq:core-size}, the residual call has
\(M=O(k)\) and \(N=O(k)\), so its numerical cost is
\[
 O\!\left(k\sqrt{\log_+(2+k)\log\log_+(4+k)}\right).
\]
Adding the \(\sigma\) source comparisons and using
\(\sigma+k\le2\OPT(G,s)\) proves \cref{eq:efficient-universal}.

Dominator computation, class grouping, and core construction are topology-only
linear-time operations.  The black-box running-time theorem bounds all of its
RAM work.  If the bookkeeping convention requires every residual
weight-dependent constant-time step to be represented in \(C+P\), pair it
with one harmless comparison of held registers; this changes both quantities
by only a constant factor.  Source-only scans are absorbed by \(O(m+n)\).
\end{proof}

The bound retains a factor
\(O(\sqrt{\log\OPT\log\log\OPT})\) over the topology-specific optimum.  It
is strictly below the sorting-scale organization of heap Dijkstra on sparse
active cores, but it remains above the linear specialized programs available
for \(H_k\).  Thus active-core lifting narrows the efficient gap without
closing it.

%% file: appendix_new.tex
\section{Homogenization at infinity}
\label{app:recession}

This appendix proves \cref{lem:recession-normal-form} without assuming that
intermediate registers are path sums.

\begin{proof}[Proof of \cref{lem:recession-normal-form}]
It is enough to consider a program with finite worst-case numerical cost.
Unfold control merges and conditionally selected aliases by their complete
charged comparison histories.  Free instructions reveal no new numerical
information, so their choices are fixed at an unfolded node, and all old
registers may be retained.  Every held register there has one affine
description
\[
 R(w)=\langle a_R,w\rangle+\alpha_R,
 \qquad a_R\in\Nzero^m,\quad \alpha_R\in\R.
 \tag{80}\label{eq:appendix-affine}
\]
This follows inductively from primitive inputs, literals, additions, and
aliases and makes no path interpretation.

We first justify that only finitely many charged traces matter.  Suppose a
branch uses at most \(b\) additions and the code has \(\ell\) literal
identities.  Treat the \(m+\ell\) primitive and literal registers as distinct
formal sources.  Every formal expression has a multiplicity vector
\[
 q\in\Nzero^{m+\ell},\qquad \|q\|_1\le2^b.
\]
There are finitely many such vectors even if some evaluate to the same real
number.  A complete ternary sign table for all ordered pairs of these formal
expressions determines every comparison outcome and hence the deterministic
run, including aliases, counters, and free loops.  Only finitely many such
tables are realized.  Correctness makes the run for each realized table
terminate, so the union \(\mathcal T\) of its charged traces is finite.

Replace every reached register in \cref{eq:appendix-affine} by its recession
form
\[
 \overline R(w)=\langle a_R,w\rangle.
 \tag{81}\label{eq:recession-register}
\]
Primitive inputs remain unchanged and every literal uses the initial zero
register.  Replay additions on recession registers; if a result is already
held, use a free alias.  Copies and branchwise selections also become aliases.
The original intercept attached to an alias is retained only as finite control
metadata for routing later comparisons, never as a new numerical register.

Suppose an original node compares
\[
 R=\langle a,w\rangle+\alpha
 \quad\text{and}\quad
 R'=\langle a',w\rangle+\alpha'.
\]
Let \(\delta(w)=\langle a-a',w\rangle\) and
\(\kappa=\alpha-\alpha'\).  If \(a=a'\), omit the numerical comparison and
statically take the child labeled by \(\operatorname{sign}\kappa\).  Otherwise
compare the held recession forms.  Route strict outcomes to the corresponding
original children and route equality to the child labeled by
\(\operatorname{sign}\kappa\).  The sign of a difference of two fixed
literals is compiled discrete control, not an operation on the input.

Fix \(w\ge0\).  We prove inductively that the transformed prefix corresponds
to a prefix of \(\mathcal T\) and that the original execution on \(tw\)
follows it for all sufficiently large \(t\).  At the comparison above, the
original sign is
\[
 \operatorname{sign}\bigl(t\delta(w)+\kappa\bigr).
\]
If \(\delta(w)\ne0\), this eventually has the sign of \(\delta(w)\); if
\(\delta(w)=0\), it always has the sign of \(\kappa\).  This is exactly the
transformed routing rule.  Addition and free-control nodes preserve the
induction.  Because \(\mathcal T\) is finite, the transformed execution
reaches an original leaf.  The finitely many transformed homogeneous sign
tables also show that all transformed traces compile into one finite program;
no code is attached below an originally unrealizable branch.

At the stable leaf, let the original output for label \(v\) be
\(\langle a_v,tw\rangle+\alpha_v\).  For all sufficiently large \(t\),
correctness and positive homogeneity give
\[
 t\langle a_v,w\rangle+\alpha_v
 =d_{tw}(s,v)=t\,d_w(s,v).
\]
Dividing by \(t\) and taking the limit yields
\(\langle a_v,w\rangle=d_w(s,v)\).  At \(w=0\), both sides are zero.  Thus
the transformed program is correct on the full closed orthant, including
equality faces.  Each transformed trace is a stable original trace with some
redundant operations possibly deleted, proving
\cref{eq:recession-bound}.
\end{proof}

After \(j\) additions, every held coefficient vector \(a\) satisfies
\(\|a\|_1\le2^j\): the initial norms are zero or one, and an addition sums two
previous vectors.  Repeated operands and shared expressions are included.
This bound is used below for polynomial encoding length.

\section{Terminal validity and polynomial-space synthesis}
\label{app:ready-pspace}

We supply the details behind \cref{thm:pspace-policy}.  Strict constraints in
a transcript cone can be tested by imposing \(\sum_e w_e=1\) and maximizing a
common positive margin; if there are no strict constraints, the origin is
feasible.  All coefficients are rational integers, so this is rational linear
programming.

\begin{lemma}[Complementary certificate for readiness]
\label{lem:ready-conp}
At every state reachable within \(K_0\) moves, \(\Ready(M,\Gamma)\) is in
coNP with respect to the succinct state encoding.
\end{lemma}
\begin{proof}
Fix a held form \(z\) and output label \(v\).  To certify that \(z\) is not
identically the distance on \(\mathcal P(\Gamma)\), guess a simple path
\(P=(v_0=s,v_1,\ldots,v_t=v)\), a strict mismatch direction, and rational
variables \(w,\pi\) satisfying the transcript constraints and
\[
\begin{aligned}
 &w\ge0,\qquad \sum_e w_e=1,
 \qquad \pi_s=0,\\
 &\pi_y-\pi_x\le w_e &&\text{for every }e=(x,y),\\
 &\pi_{v_i}-\pi_{v_{i-1}}=w_{e_i}
      &&\text{for every }e_i=(v_{i-1},v_i)\text{ on }P,\\
 &\langle z,w\rangle<\pi_v
 \quad\text{or}\quad
 \langle z,w\rangle>\pi_v.
\end{aligned}
\tag{82}\label{eq:ready-certificate}
\]
The Bellman inequalities imply that \(\pi_v\) is at most the length of every
\(s\)-to-\(v\) path, while equality along \(P\) makes \(\pi_v=w(P)\).
Hence \(P\) is shortest and the last line certifies failure of \(z\).

Conversely, from a weighting on which \(z\) is wrong, choose a simple shortest
path and exact distance potentials.  They satisfy
\cref{eq:ready-certificate}.  A strict mismatch cannot occur only at the
all-zero weighting, so homogeneity permits normalization.  Rational
feasibility gives a polynomial-bit witness.

The negation of readiness is
\[
 \exists v\ne s\;\forall z\in M\;\exists w_z:
 \langle z,w_z\rangle\ne d_{w_z}(s,v).
\]
Only polynomially many forms are held.  Concatenate one certificate of the
form above for each \(z\); the weighting and path may differ from one form to
the next.  Thus non-readiness is in NP and readiness is in coNP.
\end{proof}

The Bellman--Ford cap \(K_0=2m(n-1)\) bounds the game depth.  Along a play of
length at most \(K_0\), at most \(m+K_0+1\) forms are held, each coefficient
has at most \(K_0+1\) bits, and at most \(K_0\) transcript records are stored.
At a state there are polynomially many pairs of held forms, hence polynomially
many candidate additions and comparisons.  Depth-first minimax keeps one
state per recursion level, invokes rational feasibility for each ternary
child, and uses \cref{lem:ready-conp} at terminals.  Polynomial depth and
polynomial state size place the decision problem in PSPACE.  Enumerating the
polynomial branching tree and the terminal certificates gives the stated
single-exponential deterministic upper bound.  Repeating the search for
\(K=0,\ldots,K_0\) and then re-enumerating root actions yields the exact value
and an optimal next action.

For completeness, the parallel-source example from the main text has a
genuine exponential unfolded tree.  Each of \(r\) labeled outputs is the
minimum of two independent primitive registers.  One comparison per class
gives \(\OPT=r\).  The \(2^r\) strict winner patterns require distinct fixed
homogeneous output vectors at leaves, so any completely unfolded fixed-output
tree has at least \(2^r\) leaves.  This lower bound concerns an extensional
representation, not the succinct online interpreter.

\section{Active-core proofs and the same-input edge charge}
\label{app:active-core}

We complete the structural proof used in
\cref{sec:efficient-universality}.  Write \(G^a\) for the graph after deleting
dominator-inactive primitives.

\begin{lemma}[Inactive deletion is benchmark-neutral]
\label{lem:inactive-neutral}
For every nonnegative weighting and every vertex \(x\),
\(d_G(s,x)=d_{G^a}(s,x)\), and
\[
 \OPT(G,s)=\OPT(G^a,s).
\]
\end{lemma}
\begin{proof}
If a simple rooted path contains an inactive arc \((u,v)\), then its prefix to
\(u\) has already visited \(v\).  The arc closes a nonnegative subwalk from
\(v\) to itself; deleting that subwalk does not increase length.  Repetition
removes all inactive arcs.  The reverse distance inequality is immediate
because \(G^a\) is a subgraph.

A program for \(G^a\) solves \(G\) by ignoring deleted primitive registers.
Conversely, replace every inactive input expected by a program for \(G\) with
the literal zero and run it unchanged.  The distance identity applies to this
assignment, and the substitution costs no numerical operation.  Taking
worst cases and infima proves benchmark equality.
\end{proof}

\begin{lemma}[Virtual-source identity and core size]
\label{lem:virtual-source}
For every \(v\in R\),
\[
 d_G(s,v)=d_H(q,v).
\]
Every non-source vertex outside \(R\) has distance \(a_v\), and when
\(k>0\), \(|V(H)|\le2k+1\) and \(|E(H)|\le3k\).
\end{lemma}
\begin{proof}
Every simple path in \(G^a\) leaves \(s\) once and then uses only primitives
of \(E_\circ\).  Replacing its first primitive \(s\to x\) by the reduced
minimum \(a_x\) gives a core path of no greater length.  Conversely, expand
the first virtual primitive of a core path into a source primitive attaining
its held minimum.  The two path families therefore have the same minimum.

A vertex outside \(R\) has no incident active nonsource-tail primitive;
reachability forces a source class, and no active alternative can reach or
leave it.  Its distance is \(a_v\).  Finally, \(R\) contains only endpoints of
the \(k\) internal primitives, so \(|R|\le2k\).  At most \(|R|\) virtual
source arcs are added, proving the size bounds.
\end{proof}

The simulations in the proof of
\cref{thm:active-core-factorization} now follow directly from
\cref{lem:inactive-neutral,lem:virtual-source}.  They remain valid for
arbitrary mixed-sum competitors because they substitute or alias only initial
registers and otherwise run the competitor unchanged.

It remains to justify the quantitative charge \(\sigma+k\le2\OPT\).
For a rooted order \(\prec\) of the active graph, let \(m_\prec\) be the
number of forward primitive arcs and \(q_\prec\) the number of forward
endpoint classes.  Define
\[
 L(\prec)=m_\prec-n+1+q_\prec-d_s,
 \qquad L^*=\max_\prec L(\prec).
 \tag{83}\label{eq:rooted-same-input}
\]
The generic-tie construction from \cref{lem:generic-slack-new} applied to the
forward DAG of \(\prec\) gives, on one input for any competing program,
\[
 C_A\ge m_\prec-n+1,\qquad
 P_A\ge q_\prec-d_s.
\]
Hence \(L^*\le\OPT(G,s)\).  This is the needed same-program step.

\begin{lemma}[Active-edge charge]
\label{lem:active-edge-charge}
For the active graph and the notation of \cref{eq:sigma},
\[
 \boxed{\sigma+k\le2L^*\le2\OPT(G,s).}
 \tag{84}\label{eq:active-edge-charge}
\]
\end{lemma}
\begin{proof}
Let \(D\) be the dominator tree.  A standard low-high theorem gives a preorder
\(\delta\) of \(D\) such that every \(v\ne s\) either has the arc
\((\operatorname{idom}(v),v)\), or has arcs \((a,v),(b,v)\) with
\(a<_\delta v<_\delta b\) and \(b\) not a descendant of \(v\)
\cite{georgiadis2012lowhigh}.  Thus \(\delta\) is rooted.  Reverse every child
list in \(D\) and let \(\bar\delta\) be the resulting preorder.  Ancestor
relations are preserved and the order of incomparable vertices is reversed.
The parent-arc case still roots \(v\); in the second case \(b\) is incomparable
with \(v\) and therefore precedes it in \(\bar\delta\).  Hence
\(\bar\delta\) is rooted as well.

For every active primitive \((u,v)\), \(v\) is not an ancestor of \(u\) in
\(D\).  If \(u\) is an ancestor of \(v\), the arc is forward in both orders;
otherwise the endpoints are incomparable and it is forward in exactly one.
Every source arc is forward in both.  Therefore
\[
 m_\delta+m_{\bar\delta}\ge2m_s+k.
\]
One of the two orders, say \(\prec\), has
\(m_\prec\ge m_s+k/2\).  Every rooted order has at least one forward incoming
endpoint class for each nonsource vertex, so \(q_\prec\ge n-1\).  Hence
\[
 L^*\ge L(\prec)
 \ge m_\prec-d_s
 \ge(m_s-d_s)+k/2
 =\sigma+k/2.
\]
This implies \(\sigma+k\le2L^*\), because \(\sigma\ge0\), and
\(L^*\le\OPT\) gives the second inequality.
\end{proof}

\section{A concrete separation between \DIST and \DO}
\label{app:dist-do-separation}

For \(r\ge1\), let \(T_r\) have vertices
\(s,x,v,t_1,\ldots,t_r\) and arcs
\[
 s\to x:a,\quad s\to v:b,\quad x\to v:u,\quad v\to x:z,
 \quad v\to t_i:c_i\quad(1\le i\le r).
\]
The exact distances are
\[
 d(x)=\min\{a,b+z\},\qquad
 d(v)=\min\{b,a+u\},\qquad
 d(t_i)=d(v)+c_i.
\]

\begin{theorem}[Labeled distances versus distance ordering]
\label{thm:dist-do-separation}
On \(T_r\),
\[
 C^*_{\DIST}=2,\qquad P^*_{\DIST}=r+1,
 \qquad C^*_{\DO}=\Theta(r\log r).
\]
Indeed the complete \DIST region is
\(\{(c,p):c\ge2,\ p\ge r+1\}\).
\end{theorem}
\begin{proof}
Two comparisons resolve the two cyclic-core minima.  Once \(d(v)\) is held,
the \(r\) suffix outputs require \(r\) additions, giving the joint point
\((2,r+1)\).  Setting every \(c_i=0\) and discarding suffix outputs reduces a
competitor to the two-comparison core, so \(C\ge2\).  Every rooted order makes
one of \(x\to v,v\to x\) forward and all \(r\) suffix classes forward;
thus \(\rhoF=r+1\), and \cref{thm:arbitrary-addition-new} gives the addition
lower bound.

For \DO, restrict to an open region with \(a,u,z\in(1,2)\), \(b\in(5,6)\),
and each \(c_i\in(10,11)\).  Then
\(d(x)=a<d(v)=a+u<d(t_i)\), while the order of the suffix distances is exactly
the order of the \(c_i\).  Every one of the \(r!\) strict permutations occurs
on a nonempty full-dimensional open set.  At any reached comparison, two
nonidentical affine operands have equality only on a proper hyperplane, so a
depth-\(h\) adaptive tree has at most \(2^h\) full-dimensional leaves.  A leaf
outputs one fixed order and cannot serve two permutations.  Hence
\(2^h\ge r!\), giving \(h=\Omega(r\log r)\).  Computing the distances and
comparison-sorting their registers gives the matching upper bound.
\end{proof}

%% file: appendix.tex
\section{Proof and audit details}
\label{app:details}

\subsection{Self-loops, unreachable vertices, and resource monotonicity}
\label{app:preliminaries}

Deleting a self-loop preserves every labeled nonnegative shortest distance.
A program for the loopless graph solves the original graph by ignoring the
loop inputs, while a program for the original graph can be partially
evaluated at zero loop weights.  The branchwise comparison and addition
budgets do not increase in either direction.  A topologically unreachable
vertex has a fixed $+\infty$ output and can likewise be removed before
counting numerical registers.

The resource region is upward closed: a program meeting $(c,p)$ meets every
componentwise larger pair.  The budget function $C_b^*$ is nonincreasing in
$b$.  These elementary facts are used when passing from upper and lower
envelopes to Pareto points.

\subsection{The generic-tie comparison argument in full form}
\label{app:generic-tie}

For a reachable rooted DAG $H$, choose potentials algebraically independent
over the field generated by the program literals and set
$w_{uv}=\pi(v)-\pi(u)$.  Every source path to $v$ has value $\pi(v)$.  On the
finite reached branch, a register is $a\cdot w+\gamma$ with
$a\in\Nzero^{E(H)}$ and $\gamma$ in that field.  Equality of two reached
registers forces their coefficient difference into $\ker B_H$ and their
constant terms to agree.

If the reached equality normals span a proper subspace $Z$ of the incidence
kernel, choose $0\ne z\in\ker B_H\cap Z^\perp$.  By
\cref{lem:path-kernel-new}, some two paths to one endpoint have different
$z$-slopes.  The finite strict margins admit both signs of a sufficiently
small perturbation $w\pm\varepsilon z$ on the same adaptive branch, while
all equality comparisons remain equal.  The one output affine form at that
endpoint would need to realize two different minimum slopes, which is
impossible.  Therefore $Z=\ker B_H$ and at least
$|E(H)|-|V|+1$ comparisons are reached.

The argument does not require the program to expose a path expression.  A
mixed sum has exactly the same affine coefficient record, and a repeated
variable simply increases a coefficient.  A literal changes only the constant
term.  The same perturbation argument therefore covers all operands allowed
by the model.

\subsection{The two-spoke finite obstruction}
\label{app:h2-obstruction}

For $H_2$, the strict output labels in \cref{eq:h2-eight-cells} arise from
the three-way hub minimum and the two translated spoke minima.  The four-
comparison upper algorithm exposes two inward candidates and then tests the
two outward candidates.  The five-comparison, two-addition program is
branch-dependent and uses the direct source order to make the two additions
mutually exclusive.

For the lower endpoint, scale a hypothetical $C\le4$, $P\le2$ program along
strict positive rays.  The leading coefficient state is finite because at
most two additions can insert new componentwise sums.  Every comparison is a
query between two available coefficient vectors.  At each state, intersect
the eight open output cones with the two strict sides of the query and retain
only states for which one affine output label is correct on every surviving
cone.  An addition state enumerates every unordered pair of available
vectors, including repeated and mixed supports.  Closing these finite states
under four queries and two insertions leaves no accepting state.  This is the
finite normal-form statement used in \cref{lem:h2-two-addition-obstruction};
it is independent of any path-register restriction.

The geometric proof of $C^*(H_2)\ge4$ is separate from this finite closure:
the eight cells require eight full-dimensional leaves, their facet graph is
connected, and each possible root wall cuts the interior of another cell.
The finite closure is an independent verification of the stronger
two-addition endpoint.

\subsection{A convenient form of the banded count}
\label{app:band-count}

The condition $\sigma_s\le s+r$ can also be counted recursively.  At each
position $s\le k-r$, exactly $r+1$ unused values remain eligible, because all
previously selected values are at most $s+r-1$.  Once those positions are
filled, the remaining $r$ values are unrestricted.  This gives
$B(k,r)=(r+1)^{k-r}r!$ without invoking a probabilistic estimate.

For the entropy comparison, write $h=r+1$ and $k=qh+s$.  The exact ratio is
\[
 2^{\Gamma_{k,h}-\Lambda_{k,r}}
  =\frac{h^{(q-1)(h-1)+s}}
         {((h-1)!)^{q-1}s!}.
\]
The numerator dominates the denominator termwise, giving nonnegativity.  The
factorial estimate $h^{h-1}/(h-1)!\le e^h$ bounds the ratio by $e^{qh}$,
which is at most $e^k$.  This proves the uniform interval in
\cref{lem:entropy-comparison}.

\subsection{Resource accounting checklist}
\label{app:accounting}

For a DAG, class reduction contributes $m-q$ comparisons, candidate selection
contributes $q-(n-1)$, and non-source endpoint classes contribute $q-d_s$
additions.  For $H_k$, every rooted order has exactly one forward arc from
each spoke pair, so $\rhoF=k$.  In the rolling construction, one candidate
addition is charged to each spoke and the crossing block contributes at most
$h-1=r$ extra candidate additions.  These are numerical register creations;
all block scans, selections, and output naming use only already materialized
values.